\documentclass{llncs}

\usepackage{hyperref}

\usepackage{amsmath}
\usepackage{stmaryrd}
\usepackage{amssymb}
\usepackage{float}

\def\Po#1#2{\Pos_{#1}{(#2)}}
\def\regw#1{\Reg{(#1)}}
\DeclareMathOperator{\Reg}{RegExp}
\DeclareMathOperator{\E}{E}
\DeclareMathOperator{\h}{H}
\DeclareMathOperator{\f}{F}
\DeclareMathOperator{\G}{G}
\DeclareMathOperator{\First}{First}
\DeclareMathOperator{\Follow}{Follow}
\DeclareMathOperator{\rooot}{root}
\DeclareMathOperator{\Pos}{Pos}
\def\b#1{\overline{#1}}

\usepackage{tikz}
\usetikzlibrary{automata}
\usetikzlibrary{shadows}

\usetikzlibrary{arrows}
\usetikzlibrary{shapes}
\usetikzlibrary{decorations.pathmorphing}
\usetikzlibrary{fit}
\tikzstyle{every picture}=[>=stealth',shorten >=1pt,node distance=1.44cm,bend angle=45,initial text=,every state/.style={inner sep=0.75mm, minimum size=1mm},font=\scriptsize]

\begin{document}

\title{An Efficient Algorithm for the Equation Tree Automaton \emph{via} the $k$-C-Continuations}
\author{Ludovic Mignot, Nadia Ouali Sebti and Djelloul Ziadi  \thanks{\email{\{ludovic.mignot, nadia.ouali-sebti, djelloul.ziadi\}@univ-rouen.fr}}}
\institute{Laboratoire LITIS - EA 4108 Universit\'e de Rouen, Avenue de l'Universit\'e \\76801 Saint-\'Etienne-du-Rouvray Cedex.}

\maketitle

\begin{abstract}
  Champarnaud and Ziadi, and Khorsi \emph{et al.} show how to compute the equation automaton of word regular expression $\E$ \emph{via} the $k$-C-Continuations.
  Kuske and Meinecke extend the computation of the equation automaton to a regular tree expression $\E$ over a ranked alphabet $\Sigma$ and produce a $O(R\cdot|\E|^2)$ time and space complexity algorithm, where $R$ is the maximal rank of a symbol  occurring  in $\Sigma$ and $|\E|$ is the size of $\E$. In this paper, we give a full description of the algorithm based on the acyclic minimization of Revuz. Our algorithm, which is performed in an $O(|Q|\cdot|\E|)$ time and space complexity, where $|Q|$ is the number of states of the produced automaton, is more efficient than the one obtained by Kuske and Meinecke.
\end{abstract}

\section{Introduction}

Regular expressions, which are finite representatives of potentially infinite languages, are widely used in various application areas such as XML Schema Languages~\cite{xml}, logic and verification~\cite{verif}, \emph{etc.} The concept of word regular expressions  has been extended to tree regular expressions. Similarly to word expressions, one can convert them into  finite recognizers, the tree automata.
 
 The study of the different  ways of conversion of regular expressions into automata and \emph{vice versa} is a very  active  field. There exists a lot  of techniques to transform regular expressions (resp. regular tree expressions)  into finite automata~\cite{Brug,glushkov,khorsi,ZPC} (resp. into finite tree automata~\cite{automate2,lata}). As far as tree automata are concerned, computation algorithms are extensions of word cases. In~\cite{lata}, the computation of the position tree automaton from  a regular tree expression has been achieved by extending the classical notions of Glushkov functions defined  in~\cite{glushkov}, leading to the computation of an automaton which number of states is linear w.r.t. the number of  occurrences of symbols but which number of transitions can be exponential. In the same paper, it is proved that this automaton can be reduced  into a quadratic size recognizer.
 
 On the other side, Kuske and Meinecke have extended the notion of  word partial derivatives~\cite{antimirov} into  tree partial derivatives. They also present how to compute them extending  from words to trees~\cite{automate2} the $k$-C-Continuation algorithm by Champarnaud and Ziadi~\cite{ZPC1}. They obtain an algorithm with $O(R\cdot|\E|\cdot|\E|)$ space and time complexity where $R$ is the maximal rank of a symbol  occurring  in the finite ranked alphabet $\Sigma$ and $|\E|$ is the size of the regular expression.

In this paper, we show how to extend a notion of $k$-C-Continuation in order to compute from a regular tree expression its equation tree automaton  with an $O(|\E|+|Q|\cdot|\E|)$ time and space complexity where $|Q|$ is the number of its states. This constitutes an improvement in comparison with  Kuske and Meinecke algorithm~\cite{automate2}. The paper is organized as follows: Section~\ref{sec prelim} outlines finite tree automata over ranked trees, regular tree expressions, and linearized regular tree expressions which allows the set of positions to be defined. Next, in Section~\ref{sec tree automata} the notions of derivation and partial derivative of regular expression and set of regular expressions are 
introduced. Thus the definitions of equation tree automaton and $k$-C-Continuation tree automaton associated with the regular expression $\E$ is obtained. Afterwards, in Section~\ref{sec:algo} we present our algorithm which builds the equation tree automaton with an $O(|\E|+|Q|\cdot|\E|)$ time and space complexity. Finally, Section~\ref{sec:exemple} provides a full example of our construction.

\section{Preliminaries}\label{sec prelim}

    Let $(\Sigma,\mathrm{ar})$ be  \emph{a ranked alphabet}, where $\Sigma$ is a finite set and $\mathrm{ar}$ represents the  \emph{rank} of $\Sigma$ which is a mapping from $\Sigma$ into $\mathbb{N}$. The set of symbols of rank $n$ is denoted by $\Sigma_{n}$. The elements of rank $0$ are called  \emph{constants}. A \emph{tree} $t$ over   $\Sigma$ is inductively defined as follows: $t=a,~ t=f(t_1,\dots,t_k)$ where $a$ is any symbol in  $\Sigma_0$, $k$ is any integer satisfying $k\geq 1$, $f$ is any symbol in $\Sigma_k$ and $t_1,\dots,t_k$ are any $k$ trees over $\Sigma$. We denote by $T_{\Sigma}$ the set of trees over $\Sigma$.  \emph{A tree language} is a subset of $T_{\Sigma}$. Let ${\Sigma}_{\geq 1}=\Sigma\backslash \Sigma_0$ denote the set of  \emph{non-constant symbols} of the ranked alphabet $\Sigma$. \emph{A Finite Tree Automaton} (FTA)~\cite{automate1,automate2} ${\cal A}$ is a tuple $( Q, \Sigma, Q_{T},\Delta )$ where $Q$ is a finite set of states, $Q_T \subset Q$ is the set of \emph{final states} 
and  $\Delta\subset\bigcup_{n\geq 0}(Q \times \Sigma_{n}\times Q^n)$ is the set of  \emph{transition rules}. This set is equivalent to the function $\Delta$ from  $Q^n \times \Sigma_n \rightarrow 2^Q$ defined  by $(q,f,q_1,\dots,q_n)\in \Delta\Leftrightarrow q\in \Delta(q_1,\dots,q_n,f)$. The domain of this function can be extended to  $(2^Q)^n \times \Sigma_n  \rightarrow 2^Q$  as follows: $\Delta(Q_1,\dots,Q_n,f)=\bigcup_{(q_1,\dots,q_n)\in Q_1\times\dots\times Q_n} \Delta(q_1,\dots,q_n,f)$.  Finally, we denote by $\Delta^*$ the function from  $T_{\Sigma}\rightarrow 2^Q$  defined for any tree in $T_{\Sigma}$ as follows: 
    
  \centerline{
    $\Delta^*(t)=
      \left\{
        \begin{array}{l@{\ }l}
          \Delta(a) & \text{ if } t=a, a\in \Sigma_0 \\
          \Delta(f,\Delta^*(t_1),\dots,\Delta^*(t_n)) & \text{ if } t=f(t_1,\dots,t_n), f\in {\Sigma}_n, t_1,\ldots,t_n\in T_{\Sigma}
        \end{array}
      \right.$
  }
\noindent A tree is \emph{accepted} by ${\cal A}$ if and only if $\Delta^*(t)\cap Q_T\neq \emptyset$.   \emph{The language recognized} by ${\cal L(A)}$ is the set of trees accepted by ${\cal A}$  \emph{i.e.} ${\cal L(A)}=\{t\in T_{\Sigma}\mid \Delta^*(t)\cap Q_T\neq \emptyset\}$. A state $q\in Q$ is \emph{coaccessible} if $q\in Q_T$ or if $\exists Q^\prime=\{q_1,\ldots,q_n\}\subset Q$, $f\in \Sigma_n$, $q^\prime$ a coaccessible state in $Q$ such that $q\in Q^\prime$ and $q^\prime \in \Delta(f, q_1,\ldots, q_n)$. The \emph{coaccessible part} of the automaton $\mathcal{A}$ is the tree automaton $\mathcal{A}^\prime=(Q^\prime,\Sigma, \Delta^\prime, {Q_T}^\prime)$ where $Q^\prime=\{q\in Q\mid q \mbox{ is coaccessible}\}$ and $\Delta^\prime=\{(q,f,q_1,\ldots,q_n)\in \Delta \mid \{q,q_1,\ldots,q_n\}\subset Q^\prime\}$. It is easy to show that $\mathcal{L}(\mathcal{A})=\mathcal{L}(\mathcal{A^\prime})$. 

 Let $\sim$ be an equivalence relation over $Q$. We denote by $[q]$ the equivalence class of any state $q$ in $Q$. The \emph{quotient of} $A$ w.r.t. $\sim$ is the tree automaton $A_{/\sim}=( Q_{/\sim}, \Sigma, {Q_{T}}_{/\sim},\Delta_{/\sim} )$ where: $Q_{/\sim}=\{[q]\mid q\in Q\}$, ${Q_{T}}_{/\sim}=\{[q]\mid q\in Q_T\}$, $\Delta_{/\sim}=\{([q],f,[q_1],\ldots,[q_n]) \mid (q,f,q_1,\ldots,q_n)\in \Delta\}$.
 
  For any integer $n\geq 0$, for any $n$ languages $L_1, \dots, L_n\subset T_{\Sigma}$, and for any symbol  $f\in \Sigma_n$, $f(L_1, \dots, L_n)$ is the tree language $\lbrace f(t_1, \dots, t_n)\mid t_i\in L_i\rbrace$. The \emph{tree substitution} of a constant $c$ in $\Sigma$ by a language $L\subset T_{\Sigma}$ in a tree $t\in T_{\Sigma}$, denoted by $t\lbrace c \leftarrow L\rbrace$, is the language inductively defined by $L$ if $t=c$; $\lbrace d\rbrace$ if $t=d$ where $d\in \Sigma_0\setminus\{c\}$; $f(t_1\lbrace c \leftarrow L\rbrace, \dots, t_n\lbrace c \leftarrow L\rbrace)$ if $t=f(t_1, \dots, t_n)$ with $f\in\Sigma_n$ and $t_1, \dots, t_n$ any $n$ trees over $\Sigma$.
    Let $c$ be a symbol in $\Sigma_0$. The $c$-\emph{product} $L_1\cdot_{c} L_2$ of two languages $L_1, L_2\subset T_{\Sigma}$ is  defined by $L_1\cdot_{c} L_2=\bigcup_{t\in L_1}\lbrace t\lbrace c \leftarrow L_2\rbrace \rbrace$. The \emph{iterated $c$-product} is  inductively  defined for $L\subset T_{\Sigma}$ by:  $L^{0_c}=\lbrace c \rbrace$ and $L^{{(n+1)}_c}=L^{n_c}\cup L\cdot_{c} L^{n_c}$. The $c$-\emph{closure} of $L$ is defined by  $L^{*_c}=\bigcup_{n\geq 0} L^{n_c}$.
    
A \emph{regular expression} over a ranked alphabet $\Sigma$ is inductively defined by $\E\in \Sigma_0$, $\E=f(\E_1, \cdots, \E_n)$, $\E=(\E_1+\E_2)$, $\E=(\E_1\cdot_c \E_2)$, $\E=({\E_1}^{*_c})$, where $c\in\Sigma_0$, $n\in\mathbb{N}$, $f\in\Sigma_n$ and $\E_1,\E_2 ,\dots, \E_n$ are any $n$ regular expressions over $\Sigma$. Parenthesis can be omitted when there is no ambiguity. We write $\E_1=\E_2$ if $\E_1$ and $\E_2$ graphically coincide. We denote by $\regw{\Sigma}$ the set of all regular expressions over $\Sigma$. Every regular expression $\E$ can be seen as a tree over the ranked alphabet $\Sigma\cup \{+,\cdot_c, *_c \}$ with $c\in \Sigma_0$ where $+$ and $\cdot_c$ can be seen as a symbol of rank $2$ and $*_c$ has rank $1$. This tree is the syntax-tree $T_{\E}$ of $\E$. The \emph{alphabetical width} $||\E||$ of $\E$ is the number of occurrences of symbols of $\Sigma$ in $\E$. \emph{The size} $|\E|$ of $\E$ is the size of its syntax tree $T_{\E}$. The \emph{language} $\llbracket \E\rrbracket$  \emph{
denoted by} $\E$ is inductively defined as $\llbracket c\rrbracket=\lbrace c\rbrace$, $\llbracket f(\E_1,\E_2 , \cdots, \E_n)\rrbracket= f(\llbracket \E_1 \rrbracket, \dots, \llbracket \E_n \rrbracket)$, $\llbracket \E_1+ \E_2\rrbracket=\llbracket \E_1\rrbracket\cup\llbracket \E_2 \rrbracket$, $\llbracket \E_1\cdot_{c} \E_2\rrbracket=\llbracket \E_1\rrbracket \cdot_{c}\llbracket \E_2 \rrbracket$, $\llbracket {\E_1}^{*_c}\rrbracket=\llbracket \E_1\rrbracket^{*_c}$ where  $n\in\mathbb{N}$, $\E_1,\E_2,\dots, \E_n$ are any $n$ regular expressions, $f\in\Sigma_n$  and $c\in \Sigma_0$. It is well known that a tree language  is accepted by some tree automaton if and only if it can be denoted by a regular expression \cite{automate1,automate2}.
A regular expression $\E$ defined over $\Sigma$ is  \emph{linear} if and only if every symbol of $\Sigma_{\geq 1}$ appears at most once in $\E$. Note that any constant symbol may occur more than once. Let $\E$ be a regular expression over $\Sigma$. The  \emph{linearized regular expression} $\b\E^{\E}$ in $\E$ of a regular expression $\E$ is obtained from $\E$ by marking differently all symbols of a rank  greater than or equal to $1$ (symbols of $\Sigma_{\geq 1}$). The set of  \emph{marked symbols} with symbols of $\Sigma_0$ is the ranked alphabet  containing symbols called  \emph{positions}. We denote this set by $\Po{\E}{\E}$. When there is no ambiguity we denote by $\b{\f}$ the subexpression $\b{\f}^{\E}$ with $\f$ is a subexpression  of $\E$.
The mapping $h$ is defined from $\Po{\E}{\E}$ to $\Sigma$ with $h(\Po{\E}{\E}_m)\subset \Sigma_m$ for every  $m\in \mathbb{N}$. It associates with a marked symbol $f_j\in  \Po{\E}{\E}_{\geq 1}$ the symbol $f\in \Sigma_{\geq 1}$ and for a symbol $c\in \Sigma_0$ the symbol $h(c)=c$.
We can extend the mapping $h$ naturally to  $\regw{\Po{\E}{\E}}\rightarrow\regw{\Sigma}$ by $h(a)=a$, $h(\E_1+\E_2)=h(\E_1)+h(\E_2)$, $h(\E_1\cdot_c\E_2)=h(\E_1)\cdot_c h(\E_2)$, $h(\E_1^{*_c})=h(\E_1)^{*_c}$, $h(f_j(\E_1,\dots,\E_n))=f(h(\E_1),\dots,h(\E_n))$, with $n\in\mathbb{N}$, $a\in \Sigma_0$, $f\in \Sigma_n$, $f_j\in \Po{\E}{\E}_n$ such that $h(f_j)=f$ and $\E_1,\dots,\E_n$ any regular expressions over $\Po{\E}{\E}$.

\begin{example}\label{exp}
Let $\Sigma_0=\{a,c\}$,   $\Sigma_1=\{g,h\}$, $\Sigma_2=\{f\}$ and $\Sigma=\Sigma_0\cup\Sigma_1\cup\Sigma_2$ be a ranked alphabet.  Let $\E$, $\f$, $\G$ be the three following regular expressions over $\Sigma$: $\f=((c + a)+(g(c))^{*_ c})^{*_c}$, $\G=f(a,h(c))$  and $\E=\f\cdot_c \G$. The linearized forms of $\E$ and $\G$ are: $\b{\E}^{\E}=((c + a)+(g_1(c))^{*_{c}})^{*_{c}}\cdot_{c} f_2(a, h_3(c))$, $\b{\G}^{\G}=f_1(a,h_2(c))$. The linearized form of $\G$ in $\E$ is $\b{\G}^{\E}=f_2(a,h_3(c))$. Notice that $\Po{\G}{\G}=\lbrace a,f_1, h_2\rbrace\neq \Po{\E}{\G}=\lbrace a, f_2, h_3 \rbrace$.
\end{example}

\section{Tree Automata Computations}\label{sec tree automata}

In this section, we recall how to compute from a regular expression $\E$ a tree automaton that accepts $\llbracket \E\rrbracket$. We first recall the computation of the equation automaton $\mathcal{A}_E$ of $\E$, then we define the $k$-c-continuation automaton $\mathcal{C}_E$.

\subsection{The Equation Tree Automaton}

  In \cite{automate2}, Kuske and Meinecke extend the notion of word partial derivatives \cite{antimirov} to tree partial derivatives in order to compute from a regular expression $\E$ a tree automaton recognizing $\llbracket \E\rrbracket$. Due to the notion of  ranked alphabet, partial derivatives are no longer sets of expressions, but sets of tuples of expressions.
  
Let ${\cal N}=(\E_1,\ldots,\E_n)$ be a tuple of regular expressions, $\f$ be some regular expression and $c\in\Sigma_0$. Then ${\cal N}\cdot_c\f$ is the tuple $(\E_1\cdot_c\f,\dots,\E_n\cdot_c\f)$. For ${\cal S}$ a set of tuples of regular expressions, ${\cal S}\cdot_c \f$ is the set ${\cal S}\cdot_c \f=\{{\cal N}\cdot_c\f\mid {\cal N}\in {\cal S}\}$. Finally, $\mathrm{SET}({\cal N})=\{\E_1, \cdots,\E_m\}$ and $\mathrm{SET}({\cal S})=\bigcup_{{\cal N}\in{\cal S}}\mathrm{SET}({\cal N})$.

\begin{definition}[\cite{automate2}]
  Let $\E$ be a regular expression over a ranked alphabet $\Sigma$ and $f$ be a symbol in $\Sigma_m$ with $m\geq 1$ an integer.  The set $f^{-1}(\E)$ of tuples of regular expressions is defined as follows:

	\centerline{ 
	 \begin{tabular}{r@{\ }l}
     $f^{-1}(g(\E_1, \cdots,\E_n))$ & 
     $=\left\{
       \begin{array}{ll}
         \{(\E_1, \cdots,\E_n)\} &\text { if } f=g\\
         \emptyset &\text { otherwise} 
       \end{array}
     \right.$\\
     $f^{-1}(\f+\G)$ & $= f^{-1}(\f) \cup  f^{-1}(\G)$\\
     $f^{-1}(\f\cdot_c \G)$ & 
     $=\left\{
      \begin{array}{ll}
        f^{-1}(\f)\cdot_c \G&\text { if } c\notin\llbracket\f\rrbracket\\
        f^{-1}(\f)\cdot_c \G\cup f^{-1}(\G) &\text{ otherwise} 
      \end{array}
     \right.$\\  
     $f^{-1}({\f}^{*_c})$ & $= f^{-1}({\f})\cdot_{c} {\f}^{*_c}$\\
	 \end{tabular}
	}
   
  The function $f^{-1}$ is extended to any set $S$ of regular expressions as follows:
  
  \centerline{ $f^{-1}(S)=\bigcup_{\E\in S}f^{-1}(\E)$.}
\end{definition}

The  \emph{partial derivative} of $\E$ w.r.t. a word $w\in\Sigma_{\geq 1}^*$, denoted by $\partial_w(\E)$, is the set of regular expressions inductively defined by: 
  
  \centerline{
  $\partial_w(\E)=
    \left\{
      \begin{array}{l@{\ }l}
        \{\E\} &\text{ if } w=\varepsilon\\
        \mathrm{SET}(f^{-1}(\partial_{u}(\E)))&\text{ if } w=uf, f\in \Sigma_{\geq 1},u\in\Sigma_{\geq 1}^*
      \end{array}
    \right.$
}

\noindent The partial derivation is extended to any subset $U$ of $\Sigma_{\geq 1}^*$ as 
%
by $\partial_U(\E)=\bigcup_{w\in U}\partial_w(\E)$.
Note that $\partial_{uf}(\E)=\partial_{f}(\partial_{u}(\E))=\bigcup_{\f\in \partial_u(\E)}\partial_f(\f)$.

\begin{definition}\label{def aut eq}
  Let $\E$ be a regular expression over a ranked alphabet $\Sigma$. The  \emph{Equation  Automaton} of $\E$ is the tree automaton ${\cal A_{\E}}=(Q,\Sigma,Q_T,\Delta)$ defined by $Q= \partial_{\Sigma^*_{\geq 1}}(\E)$, $Q_T=\{\E\}$, and
  
  \centerline{$\Delta=
    \left\{
      \begin{array}{c}
        \{(\f,f,\G_1,\dots,\G_m)\mid \f\in Q,f\in\Sigma_m, {m\geq 1},(\G_1,\dots,\G_m)\in f^{-1}(\f)\}\\
        \cup\   \{(\f,c)\mid~ c\in(\llbracket \f\rrbracket\cap\Sigma_0)\}\\
      \end{array}
    \right.$}
\end{definition}

\begin{theorem}[\cite{automate2}]
Let $\E$ be a regular expression and ${\cal A_{\E}}$ be the equation tree automaton associated with $\E$. Then
 ${\cal L}({\cal A_{\E}})=\llbracket \E\rrbracket$.
\end{theorem}
\subsection{The C-Continuation Tree Automaton}

In~\cite{automate2}, Kuske and Meinecke show how to efficiently compute the equation tree automaton of a regular expression \emph{via} an extension of Champarnaud and Ziadi's 
$k$-C-Continuation~\cite{ZPC1,ZPC2,khorsi}. In this section, we show how to inductively compute them. The main difference with~\cite{automate2} is that the $k$-c-continuations are here computed using alternative formulae, and not using the partial derivation. As a consequence, any symbol that appears in the expression $\E$ admits a non-empty $k$-c-continuation (\emph{e.g.} in~\cite{automate2}, there is no continuation for $g$ in $\E=a\cdot_b g(c)$). 

\begin{definition}\label{def1}
 Let $\E$ be linear. Let $k$ and $m$ be two integers such that $1\leq k\leq m$. Let $f$ be in $(\Sigma_{\E}\cap\Sigma_{m})$. The  \emph{$k$-C-continuation} $C_{f^k}(\E)$ of $f$ in $\E$ is the regular expression defined by:
 
\centerline{
  \begin{tabular}{r@{\ }l}
    $C_{f^k}(g(\E_1, \cdots,\E_m))$ & 
      $=\left\{
        \begin{array}{l@{\ }l}
          \E_k & \text { if } f=g\\
          C_{f^k}(\E_j)& \text{ if } f\in \Sigma_{\E_j}\\
        \end{array}
      \right.$\\
    $C_{f^k}(\f+\G)$ & 
    $=\left\{
      \begin{array}{l@{\ }l}
        C_{f^k}(\f)& \text{ if } f\in \Sigma_{\f}\\
        C_{f^k}(\G)& \text{ if } f\in \Sigma_{\G}\\
      \end{array}
    \right.$\\
    $C_{f^k}(\f\cdot_c \G)$ & 
    $=\left\{
      \begin{array}{l@{\ }l}
        C_{f^k}(\f)\cdot_{c} \G &\text{ if } f\in \Sigma_{\f}\\
        C_{f^k}(\G) &\text{ otherwise }   
      \end{array}
    \right.$\\
    $C_{f^k}({\f}^{*_c})$ & $=C_{f^k}({\f})\cdot_{c} {\f}^{*_c}$\\
  \end{tabular}
}

 By convention, we set $C_{\varepsilon^1}(\E)=\E$. 
\end{definition}

Let us first show the relation between partial derivation and $k$-c-continuation.

\begin{lemma}\label{lem f cgk cf}
 Let $\E$ be linear,  $n$, $m$ and $k$ be three integers such that $n,m\geq 1$, $1\leq k\leq m$, $f\in\Sigma_n$ and  $g\in\Sigma_m\cup\{\varepsilon\}$. If $f^{-1}(C_{g^k}(\E))\neq \emptyset$ then $f^{-1}(C_{g^k}(\E))=\{(C_{f^1}(\E),\dots,C_{f^n}(\E))\}$.
\end{lemma}
\begin{proof}
  By induction over the structure of $\E$. For any symbol $g\in\Sigma_p\cup\{\varepsilon\}$ and for any expression $\f$, let us set 
  $C_g(F)=(C_{g^1}(F),\dots,C_{g^p}(F))$.
  \begin{enumerate}
    \item Let us suppose that $\E=h(\E_1, \cdots,\E_m)$. Three cases have to be considered:
    \begin{enumerate}
      \item If $g=\varepsilon$, then $k=1$ and $f^{-1}(C_{g^k}(\E))=f^{-1}(\E)$. Since $f^{-1}(C_{g^k}(\E))\neq \emptyset$, $f=h$. Hence, $f^{-1}(\E)=\{(\E_1,\ldots,\E_n)\}$. Moreover, for any integer $1\leq j\leq n$, $C_{f^j}(\E)=E_j$. Consequently, $f^{-1}(C_{g^k}(\E))=\{C_{f}(\E)\}$.
      \item Let us suppose that $g\neq\varepsilon$ and $g\neq h$. Hence $C_{g^k}(\E)=C_{g^k}(\E_l)$ with $f\in \Sigma_{\E_l}$. By induction hypothesis, $f^{-1}(C_{g^k}(\E_j))=\{C_{f}(\E_l)\}$. Moreover, for any integer $1\leq j\leq n$, $C_{f^j}(\E)=C_{f^j}(\E_l)$. Consequently, $f^{-1}(C_{g^k}(\E))=\{C_{f}(\E)\}$.
      \item Let us suppose that $g\neq\varepsilon$ and $g=h$. Hence $C_{g^k}(\E)=\E_k$. Since $f^{-1}(C_{g^k}(\E))\neq \emptyset$, then $f\in\Sigma_{E_k}$. Thus, $f\neq h$. By definition, $\E_k=C_\varepsilon^1(\E_k)$. By induction hypothesis, $f^{-1}(C_\varepsilon^1(\E_k))=\{C_{f}(\E_k)\}$. Since $f\neq h$ and since $f\in\Sigma_{E_k}$, for any integer $1\leq j\leq n$, $C_{f^j}(\E)=C_{f^j}(\E_k)$. Consequently, $f^{-1}(C_{g^k}(\E))=\{C_{f}(\E)\}$.
    \end{enumerate}
    \item Suppose that $\E=E_1+E_2$. Suppose that $f\in\Sigma_{E_1}$. Then $f^{-1}(C_{g^k}(\E))=f^{-1}(C_{g^k}(\E_1))$. By induction hypothesis, $f^{-1}(C_{g^k}(\E_1))=\{C_{f}(\E_1)\}$. Finally, since for any integer $1\leq j\leq n$, $C_{f^j}(\E)=C_{f^j}(\E_1)$, it holds $f^{-1}(C_{g^k}(\E))=\{C_{f}(\E)\}$. The prove is identical whenever $f\in\Sigma_{E_2}$.
    \item Let us suppose that $\E=E_1\cdot_c \E_2$. Two cases have to be considered:
    \begin{enumerate}
      \item If $g\in \Sigma_{E_1}$, $f^{-1}(C_{g^k}(\E))=f^{-1}(C_{g^k}(\E_1)\cdot_c \E_2)$. If $f\in\Sigma_{E_1}$, then $f^{-1}(C_{g^k}(\E_1)\cdot_c \E_2)=f^{-1}(C_{g^k}(\E_1))\cdot_c \E_2$; otherwise, $f^{-1}(C_{g^k}(\E_1)\cdot_c \E_2)=f^{-1}(\E_2)$. Hence, according to induction hypothesis, either $f^{-1}(C_{g^k}(\E_1)\cdot_c \E_2)=\{(C_{f^1}(\E_1)\cdot_c \E_2,\dots,C_{f^n}(\E))\cdot_c \E_2\}$, or $f^{-1}(C_{g^k}(\E_1)\cdot_c \E_2)=\{(C_{f^1}(\E_2),\dots,C_{f^n}(\E_2))\}$. By definition, considering whether $f\in\Sigma_{E_1}$, for any integer $1\leq j\leq n$, either $C_{f^j}(\E)=C_{f^1}(\E_1)\cdot_c \E_2$ or $C_{f^j}(\E)=C_{f^1}(\E_2)$. In both of these cases, $f^{-1}(C_{g^k}(\E))=\{C_{f}(\E)\}$.
      \item If $g\in \Sigma_{E_2}$, $f^{-1}(C_{g^k}(\E))=f^{-1}(C_{g^k}(\E_2)$. By induction hypothesis, $f^{-1}(C_{g^k}(\E_2)=\{(C_{f^1}(\E_2),\dots,C_{f^n}(\E_2))\}$. Moreover, for any integer $1\leq j\leq n$, $C_{f^j}(\E)=C_{f^j}(\E_2)$. Consequently, $f^{-1}(C_{g^k}(\E))=\{C_{f}(\E)\}$.
    \end{enumerate} 
    \item Let us suppose that $\E=E_1^{*_c}$. Two cases have to be considered:
    \begin{enumerate}
      \item If $g=\varepsilon$, then $f^{-1}(C_{g^k}(\E))=f^{-1}(\E_1^{*_c})=f^{-1}(\E_1)\cdot_c \E_1^{*_c}$. By definition, $\E_1=C_{\varepsilon}(\E_1)$. Hence by induction hypothesis, $f^{-1}(C_{\varepsilon}(\E_1))\cdot_c \E_1^{*_c}=\{(C_{f^1}(\E_1)\cdot_c \E_1^{*_c},\dots,C_{f^n}(\E_1)\cdot_c \E_1^{*_c}\}$.  Moreover, for any integer $1\leq j\leq n$, $C_{f^j}(\E)=C_{f^j}(\E_1)\cdot_c \E_1^{*_c}$. Consequently, $f^{-1}(C_{g^k}(\E))=\{C_{f}(\E)\}$.
      \item Suppose that $g\neq\varepsilon$. Then $f^{-1}(C_{g^k}(\E))=f^{-1}(C_{g^k}(\E_1)\cdot_c \E_1^{*_c})$. Depending whether $c$ belongs to $\llbracket C_{g^k}(\E_1) \rrbracket$, either $f^{-1}(C_{g^k}(\E_1)\cdot_c \E_1^{*_c})=f^{-1}(C_{g^k}(\E_1)) \cdot_c \E_1^{*_c}$  or $f^{-1}(C_{g^k}(\E_1)\cdot_c \E_1^{*_c})=f^{-1}(C_{g^k}(\E_1)) \cdot_c \E_1^{*_c} \cup f^{-1}(\E_1^{*_c})$.  Since $\E_1^{*_c}=C_\varepsilon^1(\E_1^{*_c})$, it holds by induction hypothesis that either $f^{-1}(C_{g^k}(\E_1)\cdot_c \E_1^{*_c})=\{(C_{f^1}(\E_1)\cdot_c \E_1^{*_c},\dots,C_{f^n}(\E_1)\cdot_c \E_1^{*_c})\}$  or $f^{-1}(C_{g^k}(\E_1)\cdot_c \E_1^{*_c})=\{(C_{f^1}(\E_1)\cdot_c \E_1^{*_c},\dots,C_{f^n}(\E_1)\cdot_c \E_1^{*_c})\} \cup \{C_{f}(\E_1^{*_c})\}$. Finally, since  for any integer $1\leq j\leq n$, $C_{f^j}(\E)=C_{f^j}(\E_1)\cdot_c \E_1^{*_c}$, in both of these cases, $f^{-1}(C_{g^k}(\E))=\{C_{f}(\E)\}$.
    \end{enumerate}
  \end{enumerate}
 \qed
\end{proof}

\begin{proposition}\label{prop f deriv u cf}
 Let $\E$ be linear and $f\in \Sigma_n$ with $n\geq 1$. Let $u$ be a word in ${\Sigma_{\geq 1}}^*$. If $f^{-1}(\partial_u(\E))\neq \emptyset$ then  $f^{-1}(\partial_u(\E))=\{(C_{f^1}(\E),\ldots,C_{f^n}(\E))\}$.
\end{proposition}
\begin{proof}
 By recurrence over the length of $u$.  For any symbol $g\in\Sigma_p\cup\{\varepsilon\}$ and for any expression $\f$, let us set $C_g(F)=(C_{g^1}(F),\dots,C_{g^p}(F))$.
 \begin{enumerate}
   \item Let $u=\varepsilon$. Then $f^{-1}(\partial_u(\E))=f^{-1}(\E)$. By definition, $f^{-1}(\E)=f^{-1}(C_\varepsilon^1(\E))$. According to Lemma~\ref{lem f cgk cf}, $f^{-1}(C_\varepsilon^1(\E))=\{C_{f}(\E)\}$.
   \item Let $u=wg$ with $w$ a word in ${\Sigma_{\geq 1}}^*$ and $g$ a symbol in $\Sigma_m$. Then $f^{-1}(\partial_u(\E))=f^{-1}(SET(g^{-1}(\partial_u(\E))))$. According to recurrence hypothesis, it holds that $SET(g^{-1}(\partial_u(\E)))=SET(\{C_{g}(\E)\})=\{C_{g^1}(\E),\ldots,C_{g^m}(\E)\}$. By definition, $f^{-1}(\{C_{g^1}(\E),\ldots,C_{g^m}(\E)\}=\bigcup_{1\leq i\leq m} f^{-1}(C_{g^i}(\E))$. According to Lem\-ma~\ref{lem f cgk cf}, for any integer $i$ such that $f^{-1}(C_{g^i}(\E))\neq\emptyset$, it holds  $f^{-1}(C_{g^i}(\E))=\{(C_{f^1}(\E),\ldots,C_{f^n}(\E))\}$. Since $f^{-1}(\partial_u(\E))\neq \emptyset$, there exists at least one integer $i$ such that $f^{-1}(C_{g^i}(\E))\neq\emptyset$. Consequently, $\bigcup_{1\leq i\leq m} f^{-1}(C_{g^i}(\E))=\{C_{f}(\E)\}$.
 \end{enumerate}
 \qed
\end{proof}

\begin{definition}\label{def aut c cont lin}
 The automaton ${\cal \b C_{\E}}=(Q_{\b{\cal C}},\Po{\E}{\E},\{C_{{\varepsilon}^1}(\b\E)\},\Delta_{\b{\cal C}})$ is defined by  
 \begin{itemize}
 \item $Q_{\b{\cal C}}=\{C_{f^k_j}(\b\E)\mid f_j\in \Po{\E}{\E}_m,1\leq k\leq m\}\cup\{C_{{\varepsilon}^1}(\b\E)\}$,
\item $\Delta_{\b{\cal C}}=
  \left\{
    \begin{array}{c}
      \{(C_{x}(\b\E),g_i,\mathfrak{C}_{g_i})\mid  g_i\in\Po{\E}{\E}_m,m\geq 1, \mathfrak{C}_{g_i}\in {g_i}^{-1}(C_{x}(\b\E))\}\\
      \cup\{(C_{x}(\b\E),c)\mid, c\in\llbracket C_{x}(\b\E)\rrbracket\cap\Sigma_0\}\\
    \end{array}
  \right.$
\end{itemize}
where for any symbol $g_i$ in $\Po{\E}{\E}_m$, $\mathfrak{C}_{g_i}=(C_{g^1_i}(\b\E),\dots,C_{g^m_i}(\b\E))$.
\end{definition}

The following lemma illustrates the link between ${\cal \b C_{\E}}$ and ${\cal A_{\b\E}}$.

\begin{lemma}\label{lemme3}
  The coaccessible part of ${\cal \b C_{\E}}$  is equal to ${\cal A_{\b\E}}$.
\end{lemma}
\begin{proof} 
  The expression $\b\E$ is the final state of the two automata. Let us suppose now that $q$ is a coaccessible state both in ${\cal \b C_{\E}}$ and ${\cal A_{\b\E}}$. Hence, from Definition~\ref{def aut c cont lin} and from Definition~\ref{def aut eq}: 
  
  \centerline{
    \begin{tabular}{l@{\ }l}
      & there exists a transition $(q,f,q_1,\ldots,q_n)$ in ${\cal A_{\b\E}}$\\
       $\Leftrightarrow$ & $(q_1,\ldots,q_n)\in f^{-1}(q)$  \\
       $\Leftrightarrow$ & $(q_1,\ldots,q_n)=(C_f^1(\b E),\ldots,C_f^n(\b E)) \in f^{-1}(q)$ (Proposition~\ref{prop f deriv u cf})\\
       $\Leftrightarrow$ & there exists a transition $(q,f,q_1,\ldots,q_n)$ in ${\cal \b C_{\E}}$.\\
    \end{tabular}
  }
  
  Hence, the states $q_1,\ldots,q_n$ are coaccessible from $q$ by $f$ in ${\cal \b C_{\E}}$ if and only if they are in ${\cal A_{\b\E}}$. Consequently, the 
  coaccessible part of ${\cal \b C_{\E}}$  is equal to the equation tree automaton ${\cal A_{\b\E}}$.
 \qed
\end{proof}

\begin{corollary}
 The automaton ${\cal \b C_{\E}}$ accepts $\llbracket \b\E\rrbracket$.
\end{corollary}

The  \emph{C-Continuation tree automaton} ${\cal C_{\E}}$ associated with $\E$ is obtained by replacing each transition  $(C_{x}(\b\E),g_i,C_{g^1_i}(\b\E),\dots,C_{g^m_i}(\b\E))$ of the tree automaton ${\cal \b C_{\E}}$ by $(C_{x}(\b\E),h(g_i),C_{g^1_i}(\b\E),\dots,C_{g^m_i}(\b\E))$.

\begin{corollary}
  $h({\cal L}({\cal \b C_{\E}}))={\cal L}({\cal C_{\E}})=\llbracket\E\rrbracket$.  
\end{corollary}

In what follows, for any two trees $s$ and $t$, we denote by $s\preccurlyeq t$ the relation "$s$ is a subtree of $t$". Let $k$ be an integer. We denote by $\rooot(s)$ the root of any tree $s$ and by $k\mbox{-}\mathrm{child(t)}$, for a tree $t=f(t_1,\dots,t_n)$, the $k^{th}$ child of $f$ in $t$ that is root of $t_k$ if it exists. 


  Let
  $1\leq k\leq m$ be two integers and $f_j$ be a symbol in $\Po{\E}{\E}_m$. The sets $\mathrm{First}(\E)$ is the subset of $\Po{\E}{\E}$ defined by $\mathrm{First}(\E)=\{\rooot(t)\in \Po{\E}{\E} \mid t\in \llbracket \b{\E} \rrbracket\}$. The set $\Follow(\E,f_j,k)$ is the subset of $\Po{\E}{\E}$ defined by $\Follow(\E,f_j,k)=\{g_i\in \Po{\E}{\E} \mid \exists t\in \llbracket \overline{\E} \rrbracket, \exists s\preccurlyeq t, \mathrm{root}(s)=f_j, k\mbox{-}\mathrm{child(s)}=g_i\}$.
  
\begin{proposition}[\cite{lata}]\label{prop tps lin pour follow}
  The computation of all the sets $(\Follow(\E,f_j,k))_{1\leq k\leq m, f\in\Po{\E}{\E}_m}$ can be done with an $O(|\E|)$ time and space complexity. 
\end{proposition}

\begin{proposition}\label{prop2}
 Let
  $1\leq k\leq m$ be two integers and $f_j$ be a position in $\Po{\E}{\E}_m$. If $\Follow(\E,f_j,k)\neq \emptyset$ then $\Follow(\E,f_j,k)=\First(C_{f_j^k}(\overline{\E}))$. 
 \end{proposition}
\begin{proof}

  Let $\E$ be a linear regular expression over a ranked alphabet $\Sigma$, $1\leq k\leq m$ be two integers and $f$ be a symbol in $\Sigma_m$. The set $\lambda^{f}(\E,k)$ is the subset of $\Sigma_0$ defined by $\lambda^{f}(\E,k)=\lbrace c\in\Sigma_0 \mid \exists t\in\llbracket \E\rrbracket, \exists f(t_1,\ldots,t_m)\preccurlyeq t, t_k=c\rbrace$. The set $\lambda(\E)$ is the subset of $\Sigma_0$ defined by $\lambda(\E)=\bigcup_{g\in\Sigma_m,1\leq k \leq m}\lambda^{g}(\E,k)$. 
  
  Let $\E$ be a regular expression over a ranked alphabet $\Sigma$, $1\leq k\leq m$ be two integers and $f_j$ be a symbol in $\Po{\E}{\E}_m$. In~\cite{lata}, it is shown, using alternative and equivalent formulae, that the set $\Follow(\E,f_j,k)$ is equal to $\Follow(\b\E,f_j,k)$, where $\Follow(F,f_j,k)$ is the subset of $\Po{\E}{\E}$ inductively defined for any linear regular expression $F$ as follows:
  
  \centerline{
    $\Follow(a,f,k)=\emptyset$,
  }
  
  \centerline{
    $\Follow(\E_1+E_2,f,k)=
      \left\{
        \begin{array}{l@{\ }l}
          \Follow(\E_1,f,k) & \text{ if } f\in \Sigma_{E_1},\\
          \Follow(\E_2,f,k) & \text{ if } f\in \Sigma_{E_2},\\ 
        \end{array}
      \right.$
  }
  
  \centerline{
    $\Follow(\E_1 \cdot_c \E_2,f,k)=
      \left\{
        \begin{array}{l@{\ }l}
          (\Follow(\E_1,f,k)\setminus\{c\}) \cup \First(\E_2) & \text{ if } c\in\lambda^{f}(\E_1,k),\\
          \Follow(\E_1,f,k) & \text{ if } f\in \Sigma_{E_1} \wedge c\notin\lambda^{f}(\E_1,k),\\
          \Follow(\E_2,f,k) & \text{ if } f\in \Sigma_{E_2} \wedge c\in\lambda(\E_1),\\
          \emptyset & \text{ otherwise,}
        \end{array}
      \right.$
  }
  
  \centerline{
    $\Follow(\E_1^{*_c},f,k)=
      \left\{
        \begin{array}{l@{\ }l}
          \Follow(\E_1,f,k) \cup \First(\E_1) & \text{ if } c\in\lambda(\E_1),\\
          \Follow(\E_1,f,k) & \text{ otherwise,}\\ 
        \end{array}
      \right.$
  }
  
  \centerline{
    $\Follow(g(\E_1,\ldots,\E_n),f,k)=
      \left\{
        \begin{array}{l@{\ }l}
          \First(\E_k) & \text{ if } f=g,\\
          \Follow(\E_l,f,k) & \text{ if } f\in \Sigma_{E_l}.\\ 
        \end{array}
      \right.$
  }
  
 Since by definition $\Follow(\E,f_j,k)=\Follow(\b\E,f_j,k)$, let us show by induction over $\b \E$ that if $\Follow(\E,f_j,k)\neq \emptyset$ then $\Follow(\E,f_j,k)=\First(C_{f_j^k}(\overline{\E}))$. Let us set $\b\E=F$.
 
Suppose that $F=f_j(F_1,\dots,F_m)$. Hence $\Follow(F,f_j,k)=\First(F_k)$. Moreover by definition $C_{f^k_j}(F)=F_k$. Then $\Follow(F,f_j,k)=\First(C_{f^k_j}(F))$. The property is true for the base case.

  Assuming that the property holds for the subexpressions 
  of $F$. 
  \begin{enumerate}
   \item Consider that $F=g_i(F_1,\dots,F_m)$ with $f_j\neq g_i$. Then by definition $\Follow(F,f_j,k)=\Follow(F_l,f_j,k)$ with $f_j\in\Sigma_{F_l}$. By induction hypothesis, $\Follow(F_l,f_j,k)=\First(C_{f^k_j}(F_l))$. Moreover, from Definition~\ref{def1}, $C_{f^k_j}(F)=C_{f^k_j}(F_l)$. Consequently, $\Follow(F,f_j,k)=\First(C_{f^k_j}(F))$.
\item Let us consider that $F=F_1+F_2$. Suppose that $f_j\in\Sigma_{F_i}$ with $i\in\{1,2\}$. Hence $\Follow(F_1+F_2,f_j,k)=\Follow(F_i,f_j,k)$. By induction hypothesis, $\Follow(F_i,f_j,k)=\First(C_{f^k_j}(F_i))$. From Definition~\ref{def1}, $C_{f^k_j}(F_1+F_2)=C_{f^k_j}(\E_i)$. Consequently, $\Follow(F_1+F_2,f_j,k)=\First(C_{f^k_j}(F_1+F_2))$.
  \item Consider that $F=F_1 \cdot_c F_2$. Three cases may occur.
   \begin{enumerate}
     \item Suppose that $c\in\lambda^{f_j}(F_1,k)$. Then $\Follow(F_1\cdot_cF_2,f_j,k)=(\Follow(F_1,f_j,k)\setminus\{c\}) \cup \First(F_2)$. By induction hypothesis, $\Follow(F_1\cdot_c F_2,f_j,k)=(\First(C_{f^k_j}(F_1)) \setminus\{c\}) \cup \First(F_2)$. Moreover, from Definition~\ref{def1}, $C_{f^k_j}(F)= C_{f^k_j}(F_1) \cdot_c F_2 $. Since $c\in\lambda^{f_j}(F_1,k)$, then by definition of $\lambda^{f_j}(F_1,k)$, $c \in \Follow(F_1,f_j,k)$. By induction hypothesis, $\Follow(F_1,f_j,k)=\First(C_{f^k_j}(F_1))$ and then $c \in \llbracket C_{f^k_j}(F_1)\rrbracket$. Consequently, by definition, $\First(C_{f^k_j}(F_1) \cdot_c F_2 )=(\First(C_{f^k_j}(F_1))\setminus\{c\}) \cup \First(F_2)$. Therefore, $\Follow(F_1\cdot_c F_2,f_j,k)=\First(C_{f^k_j}(F_1\cdot_c F_2))$. 
      \item Consider that $c\notin\lambda^{f_j}(F_1,k)$ and $f_j\in\Sigma_{F_1}$. In this case $\Follow(F_1\cdot_c F_2,f_j,k)=\Follow(F_1,f_j,k)$. By induction hypothesis, $\Follow(F_1,f_j,k)=\First(C_{f^k_j}(F_1))$. From Definition~\ref{def1}, $C_{f^k_j}(F_1\cdot_c F_2)=C_{f^k_j}(F_1)\cdot_c F_2$. Since $c\notin\lambda^{f_j}(F_1,k)$, then by definition $c \notin \Follow(F_1,f_j,k)$. By induction hypothesis, $\Follow(F_1,f_j,k)=\First(C_{f^k_j}(F_1))$ and then $c \notin \llbracket C_{f^k_j}(F_1)\rrbracket$. Consequently, $\First(C_{f^k_j}(F_1\cdot_c F_2))=\First(C_{f^k_j}(F_1))$. Then $\Follow(F_1\cdot_c F_2,f_j,k)=\First(C_{f^k_j}(F_1\cdot_c F_2))$.      
      \item Consider that $c\in\lambda(F_1)$ and $f_j\in\Sigma_{F_2}$. In this case $\Follow(F_1\cdot_c F_2,f_j,k)=\Follow(F_2,f_j,k)$. By induction hypothesis, $\Follow(F_2,f_j,k)=\First(C_{f^k_j}(F_2))$. From Definition~\ref{def1}, $C_{f^k_j}(F_1\cdot_c F_2)=C_{f^k_j}(F_2)$. Therefore, $\Follow(F_1\cdot_c F_2,f_j,k)=\First(C_{f^k_j}(F_1\cdot_c F_2))$. 
      \end{enumerate}
   \item Consider that $F={F_1}^{*_c}$. By Definition~\ref{def1}, $C_{f^k_j}(F_1^{*_c})=C_{f^k_j}(F_1)\cdot_c F_1$. Two cases may occur.
     \begin{enumerate}
       \item Suppose that $c\in \lambda(F_1)$. In this case, $\Follow({F_1}^{*_c},f_j,k)=\Follow(F_1,f_j,k) \cup \First({F_1}^{*_c})$. By induction hypothesis, $\Follow(F_1,f_j,k)=\First(C_{f^k_j}(F_1))$. By definition, $c\in \Follow(F_1,f_j,k)$. Since by induction $\Follow(F_1,f_j,k)=\First(C_{f^k_j}(F_1))$, $c\in \First(C_{f^k_j}(F_1))$ and then $c\in \llbracket C_{f^k_j}(F_1) \rrbracket$. Consequently, $\First(C_{f^k_j}(F_1)\cdot_c F_1)=\First(C_{f^k_j}(F_1)) \cup \First(F_1)$. Consequently $\Follow({F_1}^{*_c},f_j,k)=\Follow(F_1,f_j,k)\cup \First({F_1}^{*_c})=\First(C_{f^k_j}(F_1))\cup  \First(F_1)=\First(C_{f^k_j}(F_1^{*_c}))$.
       \item Suppose that $c\notin \lambda(F_1)$. In this case, $\Follow({F_1}^{*_c},f_j,k)=\Follow(F_1,f_j,k)$. By induction hypothesis, $\Follow(F_1,f_j,k)=\First(C_{f^k_j}(F_1))$. By definition, $c\notin \Follow(F_1,f_j,k)$. Since by induction $\Follow(F_1,f_j,k)=\First(C_{f^k_j}(F_1))$, $c\notin \First(C_{f^k_j}(F_1))$ and then $c\notin \llbracket C_{f^k_j}(F_1) \rrbracket$. Consequently, $\First(C_{f^k_j}(F_1)\cdot_c F_1)=\First(C_{f^k_j}(F_1))$. Consequently $\Follow({F_1}^{*_c},f_j,k)=\Follow(F_1,f_j,k)=\First(C_{f^k_j}(F_1))=\First(C_{f^k_j}(F_1^{*_c}))$.
     \end{enumerate}
  \end{enumerate}
  \qed
 \end{proof}

\begin{proposition}\label{prop equiv g-1 first}
 Let $1\leq k\leq m$ be two integers, $f_j$ be a symbol in $\Po{\E}{\E}_m$ and $g_i$ be a symbol in $\Po{\E}{\E}$. Then $g_i^{-1}(C_{f^k_j}(\b \E))\neq\emptyset$ $\Leftrightarrow$ $g_i \in \First(C_{f^k_j}(\b \E))$.
\end{proposition}
\begin{proof}
  Let $F$ be a linear expression. Let us show by induction over the structure of $F$ that $g_i^{-1}(F)\neq\emptyset \Leftrightarrow g_i \in \First(F)$.
  \begin{enumerate}
    \item Consider that $F=g_i(F_1,\ldots,F_n)$. By definition, $\First(F)=\{g_i\}$. By definition, $g_i^{-1}(F)=\{(F_1,\ldots,F_n)\}$. Hence the two conditions are both satisfied.
    \item Consider that $F=f(F_1,\ldots,F_n)$ with $f\in\Sigma_F\setminus\{ g_i\}$. By definition, $\First(F)=\{f\}$. By definition, $g_i^{-1}(F)=\emptyset$. Hence the two conditions are both unsatisfied.
    \item If $F=F_1+F_2$, then according to~\cite{lata}, $\First(F)=\First(F_1)\cup \First(F_2)$. By definition, $g_i^{-1}(F)=g_i^{-1}(F_1)\cup g_i^{-1}(F_2)$. By induction hypothesis, for $l\in\{1,2\}$, $g_i^{-1}(F_l)\neq\emptyset \Leftrightarrow g_i \in \First(F_l)$. Consequently, $g_i^{-1}(F)\neq\emptyset$ $\Leftrightarrow$ $g_i \in \First(F_1)\vee g_i \in \First(F_2$ $)\Leftrightarrow$ $g_i \in \First(F_1)\cup\First(F_2)$.
    \item If $F=F_1\cdot_c F_2$, then according to~\cite{lata}, $\First(F)=\First(F_1)\cup (\First(F_2)\mid c\in c\in\llbracket F_1\rrbracket)$. By definition, $g_i^{-1}(F)=g_i^{-1}(F_1)\cdot_c F_2\cup (g_i^{-1}(F_2)\mid c\in\llbracket F_1\rrbracket)$. By induction hypothesis, for $l\in\{1,2\}$, $g_i^{-1}(F_l)\neq\emptyset \Leftrightarrow g_i \in \First(F_l)$. Consequently, $g_i^{-1}(F)\neq\emptyset$ $\Leftrightarrow$ $g_i \in \First(F)$.
    \item If $F=F_1^{*_c}$, then according to~\cite{lata}, $\First(F)=\First(F_1)$. By definition, $g_i^{-1}(F)=g_i^{-1}(F_1)\cdot_c F_1^{*_c}$. By induction hypothesis, $g_i^{-1}(F_1)\neq\emptyset \Leftrightarrow g_i \in \First(F_1)$. Consequently, $g_i^{-1}(F)\neq\emptyset$ $\Leftrightarrow$ $g_i \in \First(F)$.
  \end{enumerate}  
  As a direct consequence, the conditions of Proposition~\ref{prop equiv g-1 first} are equivalent.
  \qed
\end{proof}

 \begin{lemma}\label{lem}
 Let $1\leq k\leq m$ be two integers and $f_j$ be a position in $\Po{\E}{\E}_m$.
 If $\Follow(\E,f_j,k)= \emptyset$ then $C_{f_j^k}(\overline{\E})$ is not a coaccessible state in $\mathcal{C_{\E}}$.
 \end{lemma}
\begin{proof}
Let us first show that for any state $q=C_{f_j^k}(\b{\E})$, there exists a tree $t$ such that $q\in \Delta^*(t)$, where $\Delta$ is the transition function of $\mathcal{\b C_{\E}}$ (proposition \textbf{P} in the following). By definition, $\llbracket q \rrbracket$ is not empty. If there exists a constant $c\in\llbracket q \rrbracket$, then by construction $q\in\Delta^*(c)$. 
If $t=g_i(t_1,\ldots,t_n) \in \llbracket q \rrbracket$, then by definition $g_i \in \First(q)$. According to Proposition~\ref{prop equiv g-1 first}, $g_i^{-1}(C_{f^k_j}(\b \E))\neq\emptyset$. Furthermore, according to Lemma~\ref{lem f cgk cf}, $g_i^{-1}(C_{f^k_j}(\b \E))=\{(C_{g_i^1}(\b \E),\ldots,C_{g_i^n}(\b \E))\}$. Hence, the states $q_1=C_{g_i^1}(\b \E)$, $\ldots$, $q_n=C_{g_i^n}(\b \E)$ are coaccessible from $q$. By induction hypothesis, there exists a tree $t'_l$ in $\llbracket q_l \rrbracket$ such that $q_l\in \Delta^*(t'_l)$. As a direct consequence, $q\in\Delta^*(g_i(t'_1,\ldots,t'_n))$.

Let us show that if $q$ is coaccessible, then there exists a tree $t$ in $\mathcal{L}(\mathcal{\b C_{\E}})$ for any tree $s$ satisfying $q\in \Delta^*(s)$ such that $s\preccurlyeq t$ (proposition \textbf{P'} in the following). 
  If $q=C_{\varepsilon}^1(\b\E)$, any tree $s$ such that $s\in \Delta^*(t)$ is accepted since $q$ is final. Setting $t=s$, property holds. 
  Otherwise, $q$ is coaccessible from a state $p$. By construction, there exists a transition $(p,f_j,(q_1,\ldots,q_m))$ with $q_k=q$. By induction hypothesis, there exists a tree $t$ in $\mathcal{L}(\mathcal{\b C_{\E}})$ for any tree $s'$ satisfying $p\in \Delta^*(s')$ such that $s'\preccurlyeq t$. Since any tree $s$ satisfying $q\in \Delta^*(s)$, is a subtree of a tree $s'$ satisfying $p\in \Delta^*(s')$ which root is $f_j$, there exists a tree $t$ in $\mathcal{L}(\mathcal{\b C_{\E}})$ for any tree $s$ satisfying $q\in \Delta^*(s)$ such that $s\preccurlyeq t$.

Suppose that $q=C_{f_j^k}(\b{\E})$ is a coaccessible state in $\mathcal{\b C_{\E}}$. According to \textbf{(P')}, there exists a tree $t$ in $\mathcal{L}(\mathcal{\b C_{\E}})$ for any tree $s$ satisfying $q\in \Delta^*(s)$ such that $s\preccurlyeq t$. By construction, since $C_{f_j^k}(\b{\E})$ is coaccessible by the symbol $f_j$, there exists a tree $s'\preccurlyeq t$ such that $\mathrm{root}(s')=f_j$ and $\mathrm{k-child}(s')=\mathrm{root}(s)$. By definition $\mathrm{root}(s)\in \Follow(\b\E,f_j,k)$.
As a direct consequence, if $C_{f_j^k}(\b{\E})$ is a coaccessible state in $\mathcal{C_{\E}}$, it is in $\mathcal{\b C_{\E}}$. As previously shown, this implies that $\Follow(\b\E,f_j,k)\neq\emptyset$. As a conclusion, by definition, $\Follow(\b\E,f_j,k)=\Follow(\E,f_j,k)\neq\emptyset$.
  \qed
  
\end{proof}


\subsection{From $k$-C-Continuation Automaton to Equation Automaton}

The equation automaton is a quotient of the C-Continuation one
w.r.t. the equivalence relation denoted by $\sim_e$ over the set of states of ${\cal \b C_{\E}}$ defined for any two states $q_1=C_{f^k_j}(\E)$ and $q_2=C_{g^p_i}(\E)$ by $q_1 \sim_e q_2 \Leftrightarrow h(q_1)=h(q_2)$.

\begin{proposition}\label{prop coacc part quot eq}
The coaccessible part of the finite tree automaton ${\cal C_{\E}}\diagup_{\sim_e}$ is isomorphic to the equation tree automaton ${\cal A_{\E}}$.
\end{proposition}
\begin{proof}

Let $E$ be a regular expression over an alphabet $\Sigma$. We define the inverse function of $h$ denoted by $h^{-1}: \Sigma\rightarrow 2^{\Po{\E}{\E}}$ such that for any symbol $x$ in $\Sigma$, $h^{-1}(x)=\{y\in\Po{\E}{\E}\mid h(y)=x\}$. 
 
\begin{theorem}[\cite{automate2}]\label{theo aut eq}
  Let $\E$ be a regular expression over an alphabet $\Sigma$. Then for every $u\in \Sigma^*_{\geq 1}$, $\partial_u(\E)=\bigcup_{\b u\in h^{-1}(u)}h(\partial_{\b u}(\b\E))$.
\end{theorem}

\begin{proposition}[\cite{automate2}]\label{prop aut eq}
  Let $\E$ be a regular expression over a ranked alphabet $\Sigma$. Then we have for every $f\in\Sigma_{\geq 1}$,
  
  \centerline{
    $f^{-1}(\E)=\bigcup_{f_j\in h^{-1}(f)}h({f_j}^{-1}(\b\E)) \mbox{ and }{\partial_{f}}(\E)=\displaystyle\bigcup_{f_j\in h^{-1}(f)}h({\partial_{f_j}}(\b\E))$ 
  }
\end{proposition}

  Let us denote by $\mathfrak{C}$ the coaccessible part of the finite tree automaton ${\cal C_{\E}}\diagup_{\sim_e}$.
  
  Let $Q$ be the set of states of ${\cal A_{\E}}$ and $Q'$ be the set of states of $\mathfrak{C}$. The isomorphism of the sets of states can be shown by the function $\phi: Q'\rightarrow Q: [C_{f^k_j}(\E)]\mapsto h(C_{f^k_j}(\E))$. Indeed, according to Lemma~\ref{lemme3}, $C_{f^k_j}(\E) \in \partial_{\b u}(\b\E)$ for some $\b u\in{\Po{\E}{\E}^*}_{\geq 1}$. Using Theorem~\ref{theo aut eq}, $h( C_{f^k_j}(\E))=h(\partial_{\b u}(\b\E))\in \partial_{{\Sigma^*}_{\geq 1}}(\E)=Q$. Injectivity of $\phi$ can be shown  directly from the definition of the equivalence relation $\sim_e$. For surjectivity, it is deduced from the Theorem~\ref{theo aut eq}.
  
  By definition, $\phi([C_{{\varepsilon}^1}(\E)])=C_{{\varepsilon}^1}(\E)=\E$. Hence the image of the final state of $\mathfrak{C}$ is the final state of ${\cal A_{\E}}$.
  
  Let us show that the sets of transitions are also isomorphic.
  
  Let $([C_{{f_j}^k}(\E)],g,[C_{{g^1_i}}(\E)]\ldots,[C_{{g^m_i}}(\E)])$ be a transition in $\mathfrak{C}$. Equivalently by construction, there exists a symbol $g_i$ such that $(C_{{f_j}^k}(\E),g_i,C_{{g^1_i}}(\E)\ldots,C_{{g^m_i}}(\E))$ is a transition in the accessible part of the automaton $\b{\mathcal{C}}_{\E}$. As the coaccessible part of $\b{\mathcal{C}}_{\E}$ and $\mathcal{A_{\b\E}}$ are equal (by Lemma~\ref{lemme3}), the transition $(\b\f,g_i,\b{\h}_1,\ldots,\b{\h}_m)$ is in the automaton $\mathcal{A_{\b\E}}$ with $\b\f=C_{{f_j}^k}(\E)\mbox{ and }\b{\h}_l=C_{{g^l_i}}(\E)$ for $1\leq l\leq m$; consequently $(\b{\h}_1,\ldots,\b{\h}_m)\in {g_i}^{-1}(\b\f)$.  From Proposition~\ref{prop aut eq}, it is equivalent to $(\h_1\ldots,\h_m)\in g^{-1}(\f)$. Thus $(\f,g,\h_1\ldots,\h_m)=(h(C_{{f_j}^k}(\E)),g,h(C_{{g^1_i}}(\E)),\ldots,h(C_{{g^m_i}}(\E)))$ is a transition in the automaton $\mathcal{A_{\E}}$. 
  
  Since only equivalences are stated, $([C_{{f_j}^k}(\E)],g,[C_{{g^1_i}}(\E)]\ldots,[C_{{g^m_i}}(\E)])$ is a transition in $\mathfrak{C}$ if and only if $(\phi([C_{{f_j}^k}(\E)]),g,\phi([C_{{g^1_i}}(\E)]),\ldots,\phi([C_{{g^m_i}}(\E)]))$ is a transition in $\mathcal{A_{\E}}$. 
   
  Finally, for $c\in \Sigma_0$, $([C_{{f_j}^k}(\E)],c)$ is a transition in $\mathfrak{C}$ if and only if $c\in\llbracket h(C_{{f_j}^k}(\E))\rrbracket$. Furthermore, it holds by construction that $(\phi([C_{{f_j}^k}(\E)]),c)$ is a transition in $\mathcal{A_{\E}}$ if and only if $c\in\llbracket \phi([C_{{f_j}^k}(\E)]) \rrbracket$. Consequently, $([C_{{f_j}^k}(\E)],c)$ is a transition in $\mathfrak{C}$ if and only if $(\phi([C_{{f_j}^k}(\E)]),c)$ is a transition in $\mathcal{A_{\E}}$.
  
  As a conclusion, $\phi$ is an isomorphism between $\mathfrak{C}$ and $\mathcal{A_{\E}}$.
 \qed
\end{proof}

\section{Construction of the equation tree automaton ${\cal A_{\E}}$}\label{sec:algo}

In~\cite{automate2}, the computation of the $k$-C-Continuations requires a preprocessing step which is the identification of subexpression of $\E$ in $O(|\E|^2 )$ time and space complexity. We propose an algorithm for the computation of the set of states with an $O(|\E|)$ time and space complexity.

\subsection{Computation of the set of states $Q_{\b{\cal C}}\diagup_{\sim_e}$}

The main idea is to efficiently compute the quotient $\b C_{\E}\diagup_{\sim_e}$ by converting the syntax tree into a finite acyclic deterministic word automaton. 

Let $T_{\E}$ be the syntax tree associated with 
$\E$.
The set of nodes of $T_{\E}$ is written as $\mathrm{Nodes}(\E)$. For a node $\nu$ in $\mathrm{Nodes}(\E)$, $\mathrm{sym}(\nu)$, $\mathrm{father}(\nu)$, $\mathrm{son}(\nu)$, $\mathrm{right}(\nu)$ and $\mathrm{left}(\nu)$ denote respectively the symbol, the father, the son, the right son and the left son of the node $\nu$ if they exist. We denote by $\E_{\nu}$ the subexpression rooted at $\nu$; In this case we write $\nu_{\E}$ to denote the node associated to $\E_{\nu}$. Let $\gamma:~\mathrm{Nodes}(\E)\cup\{\bot\}\rightarrow ~\mathrm{Nodes}(\E)\cup\{\bot\}$ be the function defined by: 

\centerline{
    $\gamma(\nu)= \left\{
    \begin{array}{l@{\ }l}
    \mathrm{father}(\nu) & \mbox{ if } \mathrm{sym}(\mathrm{father}(\nu))=^{*_c} \mbox{ and }\nu\neq \nu_{\E}\\
   \mathrm{right}(\mathrm{father}(\nu)) & \mbox{ if } \mathrm{sym}(\mathrm{father}(\nu))=\cdot_c\\
    \bot & \mbox{ otherwise }\\
    \end{array}\right.$
}

\noindent where $\bot$ is an artificial node such that $\gamma(\bot)=\bot$. The ZPC-Structure is the syntax tree equipped with $\gamma(\nu)$  
links.
We extend the relation $\preccurlyeq$ to the set of nodes of $T_{\E}$: For two nodes $\mu$ and $\nu$ we write $\nu\preccurlyeq\mu\Leftrightarrow T_{\E_{\nu}}\preccurlyeq T_{\E_{\mu}}$. We define the set $\Gamma_{\nu}(\E)=\{\mu\in\mathrm{Nodes}(\E)\mid \nu \preccurlyeq \mu \land \gamma(\mu)\neq \bot\}$ which is totally ordered by $\preccurlyeq$.

\begin{proposition}\label{b}
Let $\E$ be linear, $1\leq k\leq n$ be two integers and $f$ be in $\Sigma_{\E}\cap\Sigma_n$. Then $C_{f^k}(\E)= ((((\E_{\nu_0}\cdot_{op(\nu_1)} \E_{\gamma(\nu_1)})\cdot_{op(\nu_2)} \E_{\gamma(\nu_2)})
\dots \cdot_{op(\nu_{m})} \E_{\gamma(\nu_m)})$ where $\nu_f$ is the node of $T_{\E}$ labelled by $f$, $\nu_0$ is the $k\mbox{-}\mathrm{child}(\nu_{f})$, $\Gamma_{\nu_{f}}(\E)=\{\nu_1, \dots,\nu_m \}$ and for $1\leq i\leq m,~op(\nu_i)=c$ such that $\mathrm{sym}(\mathrm{father}(\nu_i)) \in \{\cdot_c,{*_c}\}$.
\end{proposition}
\begin{proof}
  By induction over the structure of $E$.
  \begin{enumerate}
    \item Let us suppose that $E=f(\E_1,\ldots,\E_n)$. Then $C_{f^k}(\E)=E_k$. Since by definition $\nu_f$ is the root of $T_E$, $k\mbox{-}\mathrm{child}(\nu_{f})$ is the root of $E_k$. Hence $\E_{\nu_0}=E_k=C_{f^k}(\E)$.
    \item Let us suppose that $E=g(\E_1,\ldots,\E_m)$ with $g\neq f$, or $E=E_1+E_2$, or $E=E_1 \cdot_c \E_2$ with $f\in \Sigma_{E_2}$. Then $C_{f^k}(\E)=C_{f^k}(\E_j)$ with $f\in \Sigma_{E_j}$. By induction hypothesis, $C_{f^k}(\E_j)=((((\E_{\nu_0} \cdot_{op(\nu_1)} \E_{\gamma(\nu_1)})\cdot_{op(\nu_2)} \E_{\gamma(\nu_2)}) \dots \cdot_{op(\nu_{m})} \E_{\gamma(\nu_m)})$ where $\nu_f$ is the node of $T_{\E_j}$ labelled by $f$, $\nu_0$ is the $k\mbox{-}\mathrm{child}(\nu_{f})$, $\Gamma_{\nu_{f}}(\E_j)=\{\nu_1, \dots,\nu_m \}$ and for $1\leq i\leq m,op(\nu_i)=c$ such that $\mathrm{sym}(\mathrm{father}(\nu_i)) \in \{\cdot_c,{*_c}\}$. Since $T_{E_j} \preccurlyeq T_E$, $C_{f^k}(\E_j)=((((\E_{\nu_0} \cdot_{op(\nu_1)} \E_{\gamma(\nu_1)})\cdot_{op(\nu_2)} \E_{\gamma(\nu_2)}) \dots \cdot_{op(\nu_{m})} \E_{\gamma(\nu_m)})$ where $\nu_f$ is the node of $T_{\E}$ labelled by $f$, $\nu_0$ is the $k\mbox{-}\mathrm{child}(\nu_{f})$, $\Gamma_{\nu_{f}}(\E_j)=\{\nu_1, \dots,\nu_m \}$ and for $1\leq i\leq m,op(\nu_i)=c$ such that 
$\mathrm{sym}(\mathrm{father}(\nu_i)) \in \{\cdot_c,{*_c}\}$.    
    \item Let us suppose that $E=E_1 \cdot_c \E_2$ with $f\in\Sigma_{E_1}$ (resp. $E=E_1^{*_c}$). Then $C_{f^k}(\E)=C_{f^k}(\E_1)\cdot_c G$ with $G\in\{E_1^{*_c},E_2\}$. By induction hypothesis, $C_{f^k}(\E_1)= ((((\E_{\nu_0}\cdot_{op(\nu_1)} \E_{\gamma(\nu_1)})\cdot_{op(\nu_2)} \E_{\gamma(\nu_2)}) \dots \cdot_{op(\nu_{m})} \E_{\gamma(\nu_m)})$ where $\nu_f$ is the node of $T_{\E_1}$ labelled by $f$, $\nu_0$ is the $k\mbox{-}\mathrm{child}(\nu_{f})$, $\Gamma_{\nu_{f}}(\E_j)=\{\nu_1, \dots,\nu_m \}$ and for $1\leq i\leq m,op(\nu_i)=c$ such that $\mathrm{sym}(\mathrm{father}(\nu_i)) \in \{\cdot_c,{*_c}\}$. Since $T_{E_1} \preccurlyeq T_E$, by setting $H=E_{\nu_{m+1}}$ and $op(\nu_{m+1})c$, $C_{f^k}(\E_1)\cdot_c H=((((\E_{\nu_0} \cdot_{op(\nu_1)} \E_{\gamma(\nu_1)})\cdot_{op(\nu_2)} \E_{\gamma(\nu_2)}) \dots \cdot_{op(\nu_{m})} \E_{\gamma(\nu_m)}) \cdot_{op(\nu_{m+1})} \E_{\gamma(\nu_{m+1})}$ where $\nu_f$ is the node of $T_{\E}$ labelled by $f$, $\nu_0$ is the $k\mbox{-}\mathrm{child}(\nu_{f})$, $\Gamma_{\nu_{
f}}(\E)=\{\nu_1, \dots,\nu_m,\nu_{m+1}\}$ and for $1\leq i\leq m+1,op(\nu_i)=c$ such that $\mathrm{sym}(\mathrm{father}(\nu_i)) \in \{\cdot_c,{*_c}\}$.
  \end{enumerate}
  \qed
\end{proof}

\begin{corollary}\label{corol}
  Let $\E$ be linear, $f\in (\Sigma_{\E})_m$ and $k\leq m$. Then $|C_{f^k_j}(\E)|\leq |\E|^2$.
\end{corollary}

\begin{example}\label{ex calc c cont}
	  Let $\Sigma$ be the ranked alphabet such that $\Sigma_0=\{a,b\}$, $\Sigma_1=\{h\}$ and $\Sigma_2=\{f\}$. Let $\E=(f(a,a)+f(a,a))^{*_a}\cdot_a h(b)$. Then $\b\E=(f_1(a,a)+f_2(a,a))^{*_a}\cdot_a h_3(b)$. The ZPC-Structure associated with $\b\E$ is represented in Figure~\ref{fig zpc e} restricted to some $\gamma$ links. As stated in Proposition~\ref{b}, $C_{f^1_1}(\b\E)=((a\cdot_a(f_1(a,a)+f_2(a,a))^{*_a})\cdot_a h_3(b))=((\E_{\nu_0}\cdot_a \E_{\gamma_{\nu_1}})\cdot_a\E_{\gamma_{\nu_2}})$. 
	\end{example}

\noindent \begin{minipage}{0.6\linewidth}
	In order to identify the equivalent $k$-C-Continuations, we can sort them in lexicographic order. 
This can be done in $O(|\E|^3)$ time and space complexity using Paige and Tarjan's Algorithm~\cite{tarjan}. 
 This is due to the fact that the size of $k$-C-Continuations is in $O(|\E|^2)$ (by Corollary~\ref{corol}). This complexity has been improved by 
 using \emph{$k$-Pseudo-Continuations} instead of $k$-C-Continuations~\cite{ZPC1,khorsi}.
\end{minipage}
\hfill
\begin{minipage}{0.35\linewidth}
	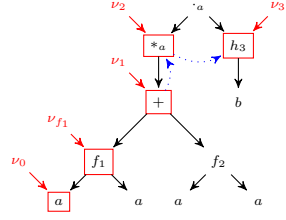
\begin{figure}[H]
	 \centerline{
	  \begin{tikzpicture}[node distance=1cm,bend angle=30,transform shape,scale=0.75]
	    \node (1) {$\cdot_a$};     
	    \node[below left of=1] (2) {$*_a$};    
	    \node[below of=2] (4) {$+$};         
	    \node[below left of=4,node distance=1.5cm] (6) {$f_1$};   
	    \node[below left of=6] (8) {$a$};     
	    \node[below right of=6] (9) {$a$};      
	    \node[below right of=4,node distance=1.5cm] (7) {$f_2$};  
	    \node[below left of=7] (10) {$a$};    
	    \node[below right of=7] (11) {$a$};        
	    \node[below right of=1] (3) {$h_3$};  
	    \node[below  of=3] (5) {$b$};
	    \node[inner sep=0pt,draw, rectangle, fit=(6), color=red] (nu) {};
	    \node[above left of=nu,node distance=1cm,color=red] (nuet) {$\nu_{f_1}$};
	    \node[inner sep=0pt,draw, rectangle, fit=(8), color=red] (nu0) {};
	    \node[above left of=nu0,node distance=1cm,color=red] (nu0et) {$\nu_0$};
	    \node[inner sep=0pt,draw, rectangle, fit=(4), color=red] (nu1) {};
	    \node[above left of=nu1,node distance=1cm,color=red] (nu1et) {$\nu_1$};
	    \node[inner sep=0pt,draw, rectangle, fit=(2), color=red] (nu2) {};
	    \node[above left of=nu2,node distance=1cm,color=red] (nu2et) {$\nu_2$};
	    \node[inner sep=0pt,draw, rectangle, fit=(3), color=red] (nu3) {};
	    \node[above right of=nu3,node distance=1cm,color=red] (nu3et) {$\nu_3$};
	    \path[->,color=red]
	      (nuet) edge (nu)
	      (nu0et) edge (nu0)
	      (nu1et) edge (nu1)
	      (nu2et) edge (nu2)
	      (nu3et) edge (nu3)
	    ;
	    \path[->,color=blue,dotted]
	      (4) edge[bend right] (2)
	      (2) edge[bend right] (3)
	    ;
	    \path[->]
	      (1) edge (2)
	      (1) edge (3)
	      (2) edge (4)
	      (3) edge (5)
	      (4) edge (6)
	      (4) edge (7)
	      (6) edge (8)
	      (6) edge (9)
	      (7) edge (10)
	      (7) edge (11)
	    ;
	  \end{tikzpicture}
	 }
	  \caption{ZPC-Structure of $E$.}
	  \label{fig zpc e}
	\end{figure}
\end{minipage}

\noindent  A \emph{$k$-Pseudo-Continuation} $l_{f^k_j}(\E)$ of $f_j$ in $\E$ is obtained from the 
 $k$-C-Continuation $C_{f^k_j}(\b\E)$ by replacing some subexpression $\b\f$ of $\b\E$ by a symbol $\psi(h(\b\f))$ such that for two subexpressions 
 $\f$ and $\G$ of $\E$:   
    $\psi(\f)=\psi(\G)\Leftrightarrow \f=\G$.

\begin{definition}\label{def2}
Let $\h$ be a regular expression over $\Sigma$ and $\psi$ be a bijection that associates to each subexpression of $\E$ a symbol in an alphabet $\Psi$. We define the word $\psi^\prime(\h)$ over the alphabet $\Psi\cup\{\cdot_a\mid a\in \Sigma_0\}$ inductively as follows:
   
 \centerline{$\psi^\prime(\h)=
      \left\{
        \begin{array}{l@{\ }l}
          \psi^\prime(\f)\cdot_c\psi(\G)& \text{ if } \h=\f\cdot_c\G \text{ and }\G \text{ a subexpression of }\E\\
           \psi(\h)& \text{ if } \h\neq\f\cdot_c\G \text{ and }\h \text{ a subexpression of }\E\\
          \varepsilon & \text{ otherwise}.
        \end{array}
      \right.$}
      
  \noindent The function $\psi^\prime$ is said to be an $(\E,\Psi)$-encoding.
\end{definition}

\begin{definition}
 Let $n$ and $k$ be two integers such that $1\leq k\leq n$, $f_j$ be a symbol in $\Po{\E}{\E}$  and $\psi^\prime$ an $(\E,\Psi)$-encoding for some alphabet $\Psi$. The $k$-Pseudo-Continuation of $f_j$ in $\E$, denoted by $l_{f^k_j}(\E)$, is the word over $\Psi\cup\{\cdot_a\mid a\in \Sigma_0\}$ defined by $l_{f^k_j}(\E)=\psi^\prime(h(C_{f^k_j}(\b\E)))$.
\end{definition}

In the following, we consider that the pseudo-continuations of 
$E$
are defined over $\Psi$ a finite subset of $\mathbb{N}$, bounded by the number of subexpressions of $E$.

\begin{lemma}\label{lem egal psi egal exp}
  Let $E$ and $F$ be two regular expressions over an alphabet $\Sigma$ such that $E$ and $F$ are two products of subexpressions of a regular expression $H$ over $\Sigma$. Let $\psi'$ be a $(H,\Psi)$-encoding. Then:
  
  \centerline{
    $\psi'(\E)=\psi'(F)$ $\Leftrightarrow$ $E=F$.
  }
\end{lemma}
\begin{proof}
  Let us consider that $\psi'$ is associated with the bijection $\psi$. Let us consider the possible roots of the expressions.
  \begin{enumerate}
    \item If the roots of $E$ and $F$ are notconcatenation products, $\psi'(\E)=\psi'(F)$ $\Leftrightarrow$ $E=F$ since $\psi$ is a bijection and $\psi'(\E)=\psi(\E)\ wedge \psi'(F)=\psi(F)$.
  
    \item Let us suppose that, without loss of generality, only the root of $E$ is aconcatenation product $\cdot_c$. Then the symbol $\cdot_c$ appears in $\psi'(\E)$ but not in $\psi'(F)$. Hence $\psi'(\E)\neq \psi'(F)$ and $E\neq F$.
  
    \item Finally, let us suppose that $E=E_1\cdot_c \E_2$ and $F=F_1\cdot_d F_2$. 
    \begin{enumerate}
      \item If $E_2\neq F_2$ then $\psi(\E)\neq \psi(F)$ and then $\psi'(\E)$ and $\psi'(F)$ do not end with the same symbol.
      \item Suppose that $E_2=F_2$. If $\cdot_c\neq \cdot_d$, $\psi'(\E)$ ends with $\cdot_c \psi(\E_2)\neq \cdot_d \psi(\E_2)$ ; Hence $\psi'(\E)\neq \psi'(F)$ and $E\neq F$. Otherwise, by induction hypothesis, $\psi'(\E_1)=\psi'(F_1)$ $\Leftrightarrow$ $E_1=F_1$. Hence $\psi'(\E_1)\cdot_c \psi(\E_2)=\psi'(F_1)\cdot_c \psi(F_2)$ $\Leftrightarrow$ $E_1\cdot_c \E_2=F_1\cdot_d F_2$.
    \end{enumerate}
  \end{enumerate}
  \qed
\end{proof}

\begin{proposition}\label{prop long pseudo cont lin}
  The two following propositions hold:
  \begin{enumerate}
    \item $|l_{f^k_j}(\E)|$ is at most linear w.r.t. $ |\E|$,
    \item $\sum_{f_j\in \Po{\E}{\E}_n,1\leq k\leq n} |\psi^{\prime}(\E_{\nu_{f^k_j}})| $ is at most linear w.r.t. $ |\E|$, with $\nu_{f^k_j}=k\mbox{-}\mathrm{child}(\nu_{f_j})$. 
  \end{enumerate} 
\end{proposition}
  \begin{proof}
  
We define the function $\mathrm{nbdot}(\E)$ as the number of 
 left-associated
 concatenation operators
 in $\E$ as follows: 

\centerline{
    $\mathrm{nbdot}(\E)=
      \left\{
        \begin{array}{l@{\ }l}
          \mathrm{nbdot}(\f)+1 & \text{ if } \E=\f\cdot_c\G,\\
          0 & \text{ otherwise. }\\ 
        \end{array}
      \right.$
  }
 
 
 Let us first prove that $|\psi^{\prime}(\E)|\leq 2\mathrm{nbdot}(\E)+1$. The proof proceeds by induction in the structure of $\E$.

 \begin{enumerate} 
   \item Whenever $\E$ is not a product, $\psi'(\E)=\psi(E)$. Hence $|\psi'(\E)|=1$. Since $\mathrm{nbdot}(\E)=0$, the condition is satisfied.
   \item Suppose that $\E=\f\cdot_c\G$. Hence 

  \centerline{
    \begin{tabular}{l@{\ }l}
      $|\psi^{\prime}(\f\cdot_c\G)|$ & $\leq |\psi^{\prime}(\f)|+2$\\
      & $\leq 2(\mathrm{nbdot}(\f))+1+2$\\
      & $= 2(\mathrm{nbdot}(\f)+1)+1$\\
      & $=2(\mathrm{nbdot}(\E))+1$.\\
    \end{tabular}
  }
  \end{enumerate}
 
  
  Following Proposition~\ref{b}, $|l_{f^k_j}(\E)|\leq | \psi'(E_{\nu_0})| +2m$ where $\nu_f$ is the node of $T_{\E}$ labelled by $f$, $\nu_0$ is the $k\mbox{-}\mathrm{child}(\nu_{f})$, $\Gamma_{\nu_{f}}(\E)=\{\nu_1, \dots,\nu_m \}$ and for $1\leq i\leq m,~op(\nu_i)=c$ such that $\mathrm{sym}(\mathrm{father}(\nu_i)) \in \{\cdot_c,{*_c}\}$. Since $\{\nu_1, \dots,\nu_m \}$ are ancestors of $\nu_0$, $|E_{\nu_0}|\leq |E|-m$. Consequently, $\mathrm{nbdot}(E_{\nu_0})\leq |E|-m$. Moreover, from previous point, $| \psi'(E_{\nu_0})|\leq 2(\mathrm{nbdot}(E_{\nu_0}))+1$. Consequently, $|l_{f^k_j}(\E)|\leq 2(|E|-m)+1 +2m=2|E|+1$.

 Furthermore, since $|\psi^{\prime}(\E_{\nu_{f^k_j}})|\leq 2(\mathrm{nbdot}(\E_{\nu_{f^k_j}}))+1$, it holds:
 
 \centerline{$\sum_{f_j\in \Po{\E}{\E}_n}\sum_{1\leq k\leq n}|\psi^{\prime}(\E_{\nu_{f^k_j}})|\leq \sum_{f_j\in \Po{\E}{\E}_n}\sum_{1\leq k\leq n} 2(\mathrm{nbdot}(\E_{\nu_{f^k_j}}))+1$.}
 
 However, the concatenation operators below the node $\nu_{f^k_j}$ do not appears below another symbol. Consequently, 
 
 \centerline{$\sum_{f_j\in \Po{\E}{\E}_n}\sum_{1\leq k\leq n} \mathrm{nbdot}(\E_{\nu_{f^k_j}}) \leq |E|$}
 
 Finally,
 
 \centerline{$\sum_{f_j\in \Po{\E}{\E}_n}\sum_{1\leq k\leq n}|\psi^{\prime}(\E_{\nu_{f^k_j}})|\leq 2|E| + |\Sigma_{\geq 1}|$.}
 
  \qed
 \end{proof}

\begin{proposition}\label{prop1}
Let $f_j\in \Po{\E}{\E}_n$, $g_i\in \Po{\E}{\E}_m$, $k\leq n$ and $p\leq m$ be two integers. Then $h(C_{f^k_j}(\b\E))= h(C_{g^p_i}(\b\E))\Leftrightarrow l_{f^k_j}(\E)=l_{g^p_i}(\E)$.
\end{proposition}
\begin{proof}
  Direct Corollary of Lemma~\ref{lem egal psi egal exp}.
 \qed
\end{proof}

  From Proposition~\ref{prop1} we can deduce that the $k$-C-Continuations identification can be achieved by considering the $k$-Pseudo-Continuations. In the following we show that this identification step (computation of ${\sim_e}$) can be done without the computation of the $k$-Pseudo-Continuations and  that it amounts to the minimization of a word acyclic deterministic automaton. Before seeing how the identification of $k$-Pseudo-Continuations $l_{f^k_j}(\E)$ is performed, we prove that the computation of the function $\psi$ can be done in a linear time in the size of $\E$. 

Let us consider the syntax tree $T_{\E}$ associated with 
$\E$.
 This syntax tree contains all the subexpressions of $\E$. Each node $\nu$ in $T_{\E}$ corresponds to the 
subexpression $\E_{\nu}$ of $\E$. The equivalence relation $\sim$ over the nodes of the tree $T_{\E}$ is defined by $\nu_1 \sim\nu_2 \Leftrightarrow \E_{\nu_1}=\E_{\nu_2}$. We show that the computation of the equivalence relation $\sim$ amounts to the minimization of the word acyclic deterministic automaton ${\cal A}_{T_{_{\E}}} =(Q,\Sigma_{\cal A},\{\nu_{\E}\},\{\nu_T\},\delta)$, where $\nu_E$ is the node associated to the root of $E$, $Q=\mathrm{Nodes}(\E)\cup \{ \nu_T \}\cup\{\bot\}$ with $\nu_T,\bot\notin \mathrm{Nodes}(\E)$, $\Sigma_{\cal A}=\Sigma_0\cup\{g_+,d_+\}\cup\{*_a,g_{\cdot_a},d_{\cdot_a}\mid a\in \Sigma_0\}\cup\{f^1,\dots,f^n\mid f\in \Sigma_n, n\geq 1\}$, and $\delta$ is defined by $\delta(\nu,*_a)=\mathrm{son}(\nu)$ if $\mathrm{sym(\nu)}=*_a$, $\delta(\nu,g_{\mathrm{sym(\nu)}})=\mathrm{left}(\nu)$ and $\delta(\nu,d_{\mathrm{sym(\nu)}})=\mathrm{right}(\nu)$ if $\mathrm{sym(\nu)}\in \{+, \cdot_a,~a\in \Sigma_0\}$, $\delta(\nu,{\mathrm{sym(\nu)}})=\nu_T$ if $\mathrm{sym(\nu)}\in\Sigma_0$, $\
delta(\nu,f^k)=k\mbox{-}\mathrm{child(\nu)}$ if $\mathrm{sym(\nu)}=f\in\Sigma_{\geq 1}$, and $\delta(\nu,x)=\bot$ in all otherwise. 

\begin{lemma}\label{lem exp egal si lang eq}
  $E=F$ $\Leftrightarrow$ $\mathcal{L}(\mathcal{A}_{T_E})=\mathcal{L}(\mathcal{A}_{T_F})$.
\end{lemma} 
\begin{proof}
  Let $\Sigma_{\mathcal{A}_E}$ (resp. $\Sigma_{\mathcal{A}_F}$) be the alphabet of the automaton $\mathcal{A}_{T_E}$ (resp. $\mathcal{A}_{T_F}$).
  
  \begin{enumerate}
    \item If $\E=\f$ then $\mathcal{A}_{T_E}=\mathcal{A}_{T_F}$ and $\mathcal{L}(\mathcal{A}_{T_E})=\mathcal{L}(\mathcal{A}_{T_F})$.
    \item Suppose that $E\neq F$.  Notice that any word $w$ in $\mathcal{L}(\mathcal{A}_{T_E})$ (resp. $\mathcal{L}(\mathcal{A}_{T_F})$) starts with a symbol associated with the root of $E$ (resp. $F$). 
    \begin{enumerate}
      \item Hence, if the roots of $E$ and $F$ are distinct, then $\mathcal{L}(\mathcal{A}_{T_E}) \cap \mathcal{L}(\mathcal{A}_{T_F})=\emptyset$ ; Since $\mathcal{L}(\mathcal{A}_{T_E})$ is not empty by construction, $\mathcal{L}(\mathcal{A}_{T_E}) \neq \mathcal{L}(\mathcal{A}_{T_F})$. 
      \item Otherwise, there exists an integer $j$ such that $E=x(\E_1,\ldots,E_k)$, $F=x(F_1,\ldots,F_k)$ and $E_j\neq F_j$. By induction over the size of $E_j$.
      \begin{enumerate}
        \item If $E_j\in\Sigma_0$, then since the roots are distincts, the word starting with the symbol associated to the node $x$ followed by the symbol $a$ is in $\mathcal{L}(\mathcal{A}_{T_{E}})$ but not in $\mathcal{L}(\mathcal{A}_{T_F})$.
        \item Otherwise, it holds by induction hypothesis that there exists a word in $\mathcal{L}(\mathcal{A}_{T_{E_j}})$ not in $\mathcal{L}(\mathcal{A}_{T_{F_j}})$. Hence there exists a word starting with a symbol associated to the node $x$ followed by a word in $\mathcal{L}(\mathcal{A}_{T_{E_j}})$ that is in $\mathcal{L}(\mathcal{A}_{T_E})$ but not in $\mathcal{L}(\mathcal{A}_{T_F})$.
      \end{enumerate}
    \end{enumerate}
  \end{enumerate}
  \qed
\end{proof}

 According to Lemma~\ref{lem exp egal si lang eq}, $\nu_1 \sim\nu_2 \Leftrightarrow \mathcal{L}(\mathcal{A}_{T_{E_{\nu_1}}})=\mathcal{L}(\mathcal{A}_{T_{E_{\nu_2}}})$, that is the equivalence relation $\sim$ coincides with Myhill-Nerode equivalence~\cite{Nerode} over the states of the automaton ${\cal A}_{T_{_{\E}}}$, that can be computed in $O(|\E|)$ time and space complexity using Revuz Algorithm~\cite{revuz}.
 
\begin{lemma}\label{lemma5}
  The computation of $\psi(\f)$ for all subexpression $\f$ of $\E$ can be done in $O(|\E|)$ time and space complexity.
\end{lemma}
\begin{proof}
  Let $\nu_1$ and $\nu_2$ be two nodes in $T_E$. As $\nu_1 \sim \nu_2 \Leftrightarrow \E_{\nu_1}=\E_{\nu_2}\Leftrightarrow \psi(\E_{\nu_1})=\psi(\E_{\nu_2})$, we can associate to each 
node $\nu$ in $T_{\E}$ (each $\E_\nu$) a symbol ($\psi(\E_\nu)$) which uniquely identifies its equivalence class $[\nu]_{\sim}$. Furthermore, according to Lemma~\ref{lem exp egal si lang eq}, the computation of the equivalence relation $\sim$ amounts to the minimization of the word acyclic deterministic automaton ${\cal A}_{T_{_{\E}}}$, which can be performed in $O(|\E|)$ using Revuz Algorithm~\cite{revuz}. 
 \qed
\end{proof}

\noindent \begin{minipage}{0.45\linewidth}
  \begin{figure}[H]
	 \centerline{
	  \begin{tikzpicture}[node distance=1cm,bend angle=30,transform shape,scale=0.75]
	    \node[state,initial above] (1) {$\cdot_a$};     
	    \node[state,below left of=1] (2) {$*_a$};    
	    \node[state,below of=2] (4) {$+$};         
	    \node[state,below left of=4,node distance=1.5cm] (6) {$f_1$};   	    \node[state,below left of=6] (8) {$a$};     
	    \node[state,below right of=6] (9) {$a$};      
	    \node[state,below right of=4,node distance=1.5cm] (7) {$f_2$};  	    \node[state,below left of=7] (10) {$a$};    
	    \node[state,below right of=7] (11) {$a$};  
	    \node[state,right of=11] (5) {$b$};           
	    \node[state,above of=5,node distance=0.75cm] (3) {$h_3$};
	    \node[accepting,state,below of=4,node distance=3cm] (12) {$\nu_T$};        
	    \path[->]
	      (1) edge node[above left,pos=0.4] {$g_{\cdot_a}$} (2)
	      (1) edge node[above right,pos=0.4] {$d_{\cdot_a}$} (3)
	      (2) edge node[left,pos=0.4] {$*_a$} (4)
	      (3) edge node[right,pos=0.4] {$h^1$} (5)
	      (4) edge node[above left,pos=0.4] {$g_+$} (6)
	      (4) edge node[above right,pos=0.4] {$d_+$} (7)
	      (6) edge node[above left,pos=0.4] {$f^1$} (8)
	      (6) edge node[above right,pos=0.4] {$f^2$} (9)
	      (7) edge node[above left,pos=0.4] {$f^1$} (10)
	      (7) edge node[above right,pos=0.4] {$f^2$} (11)
	      (8) edge node[left,pos=0.4] {$a$} (12)
	      (9) edge node[left,pos=0.4] {$a$} (12)
	      (10) edge node[left,pos=0.4] {$a$} (12)
	      (11) edge node[left,pos=0.4] {$a$} (12)
	      (5) edge node[below right,pos=0.4] {$b$} (12)
	    ;
	  \end{tikzpicture}
	 }
	  \caption{The automaton ${\cal A}_{T_{_{\E}}}$ .}
	  \label{fig a t e}
	\end{figure}
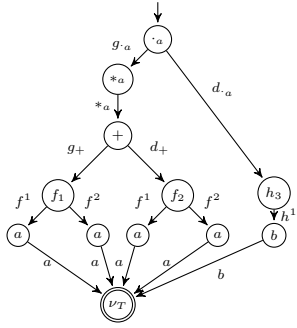
\end{minipage}
\hfill
\begin{minipage}{0.45\linewidth}
	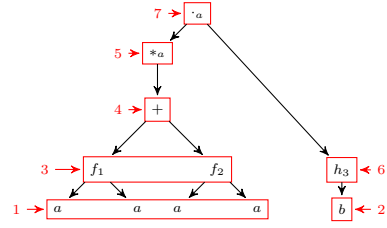
\begin{figure}[H]
	 \centerline{
	  \begin{tikzpicture}[node distance=1cm,bend angle=30,transform shape,scale=0.75]
	    \node (1) {$\cdot_a$};     
	    \node[below left of=1] (2) {$*_a$};    
	    \node[below of=2] (4) {$+$};         
	    \node[below left of=4,node distance=1.5cm] (6) {$f_1$};   
	    \node[below left of=6] (8) {$a$};     
	    \node[below right of=6] (9) {$a$};      
	    \node[below right of=4,node distance=1.5cm] (7) {$f_2$};  
	    \node[below left of=7] (10) {$a$};    
	    \node[below right of=7] (11) {$a$};  
	    \node[right  of=11, node distance=1.5cm] (5) {$b$};       
	    \node[above of=5,node distance=0.7cm] (3) {$h_3$}; 
	    \node[inner sep=0pt,draw, rectangle, fit=(8)(11), color=red] (eq1) {};
	    \node[left of=eq1,node distance=2.5cm,color=red] (eq1et) {$1$};
	    \node[inner sep=0pt,draw, rectangle, fit=(5), color=red] (eq2) {};
	    \node[right of=eq2,node distance=0.7cm,color=red] (eq2et) {$2$};
	    \node[inner sep=0pt,draw, rectangle, fit=(6)(7), color=red] (eq3) {};
	    \node[left of=eq3,node distance=2cm,color=red] (eq3et) {$3$};
	    \node[inner sep=0pt,draw, rectangle, fit=(3), color=red] (eq4) {};
	    \node[right of=eq4,node distance=0.7cm,color=red] (eq4et) {$6$};
	    \node[inner sep=0pt,draw, rectangle, fit=(4), color=red] (eq5) {};
	    \node[left of=eq5,node distance=0.7cm,color=red] (eq5et) {$4$};
	    \node[inner sep=0pt,draw, rectangle, fit=(2), color=red] (eq6) {};
	    \node[left of=eq6,node distance=0.7cm,color=red] (eq6et) {$5$};
	    \node[inner sep=0pt,draw, rectangle, fit=(1), color=red] (eq7) {};
	    \node[left of=eq7,node distance=0.7cm,color=red] (eq7et) {$7$};
	    \path[->,color=red]
	      (eq1et) edge (eq1)
	      (eq2et) edge (eq2)
	      (eq3et) edge (eq3)
	      (eq4et) edge (eq4)
	      (eq5et) edge (eq5)
	      (eq6et) edge (eq6)
	      (eq7et) edge (eq7)
	    ;
	    \path[->]
	      (1) edge (2)
	      (1) edge (3)
	      (2) edge (4)
	      (3) edge (5)
	      (4) edge (6)
	      (4) edge (7)
	      (6) edge (8)
	      (6) edge (9)
	      (7) edge (10)
	      (7) edge (11)
	    ;
	  \end{tikzpicture}
	 }
	  \caption{The Equivalence Classes.}
	  \label{fig es sim}
	\end{figure}
\end{minipage}

\begin{example}
  Let us consider the regular expression  $\E=(f(a,a)+f(a,a))^{*_a}\cdot_a h(b)$ of the Example~\ref{ex calc c cont}. Applying Myhill-Nerode equivalence~\cite{Nerode} to the states of the automaton ${\cal A}_{T_{_{\E}}}$ (Figure~\ref{fig a t e}) results in $7$ equivalence classes labeled by $\Psi=\{1,2,\ldots,7\}$. For example $\psi(f(a,a))=3$ and $\psi(\E)=7$ (Figure~\ref{fig es sim}). Finally, $l_{f^1_1}(\E)=1\cdot_a 6\cdot_a 5$. 
\end{example}

Recall that the $k$-Pseudo-Continuation identification can be achieved in $O(|\E|^2)$ \cite{ZPC2,automate2} using Paige and Tarjan's sorting algorithm~\cite{tarjan}. In what follows we show that this step amounts to the minimization of the acyclic deterministic word automaton ${\cal B}_{T_{_{\b\E}}} =(Q_{\cal B},\Sigma_{\cal B},\{\nu_T\},\{\nu_{\b\E}\},\delta_{\cal B})$ defined  with $\nu_T\notin\mathrm{Nodes}(\b\E)$ and $\mathfrak{F}=\{f^k_j \mid 1\leq k \leq m, f_j\in \Po{\E}{\E}_m\}$ by: 
\begin{itemize}
\item $Q_{\cal B}=(\mathrm{Nodes}(\b\E)\setminus\Sigma_0)\cup \mathfrak{F} \cup\{\nu_T,\bot\}$,
\item $\Sigma_{\cal B}=\{\psi(\nu)\mid \nu\in\mathrm{Nodes}(\b\E)\cap Q_{\cal B}\}\cup\mathfrak{F} \cup\{\cdot_a\mid a\in\Sigma_0\}\cup\{\varepsilon\}$,
\item $\delta_{\cal B}$ is defined as follows:
	\begin{itemize}
	\item $\delta(\nu_T,f^i_j)=f^i_j$ for all  $f^i_j\in \mathfrak{F}$,
	\item $\delta(f^k_j,\psi^\prime(h(\E_{\nu_k})))=f_j$ if $\nu_k$ is the $k^{th}$ child of $f_j$,
	\item $\delta(\nu,\cdot_a\psi(\E_{\gamma(\nu)}))=\mathrm{father}(\nu)$ if $\mathrm{sym(father(\nu))}\in\{\cdot_a,*_a\}$ and $\gamma(\nu)\neq\bot$,
	\item $\delta(\nu,\varepsilon)=\mathrm{father}(\nu)$ and if $\gamma(\nu)=\bot$ and $\delta(\nu,x)=\bot$ in all otherwise.  
	 \end{itemize}
\end{itemize}     

\begin{proposition}\label{prop lang B-T-E}
$\mathcal{L}(\mathcal{B}_{T_{\overline{\E}}})=\{f_j^k \cdot l_{f^k_j}(\E)\mid f_j\in\Po{\E}{\E}_m, k\leq m \}$
\end{proposition}
\begin{proof}
  By construction of $\mathcal{B}_{T_{\overline{\E}}}$, there exists a path from any state $f_j^k$ with $f_j\in\Po{\E}{\E}_m$ and $1\leq k\leq m$ to the root of $E$ labelled by $ \psi'(\E_{\nu_0}) \cdot_{op(\nu_1)} \psi(\E_{\gamma(\nu_1)})\cdots \cdot_{op(\nu_{m})} \psi( \E_{\gamma(\nu_m)})$ where $\nu_{f_j}$ is the node of $T_{\b\E}$ labelled by $f_j$, $\nu_0$ is the $k\mbox{-}\mathrm{child}(\nu_{f})$, $\Gamma_{\nu_{f_j}}(\E)=\{\nu_1, \dots,\nu_m \}$ and for $1\leq i\leq m,op(\nu_i)=c$ such that $\mathrm{sym}(\mathrm{father}(\nu_i)) \in \{\cdot_c,{*_c}\}$. This word exactly corresponds to the word $\psi'(h(C_{f^k_j}(\b\E)))=l_{f^k_j}(\E)$.
 \qed
\end{proof} 

Let $f_j$ and $g_i$ be two positions in $\Po{\E}{\E}$.  As a direct consequence of Proposition~\ref{prop lang B-T-E}, $C_{f^k_j}(\b\E)\sim_e C_{g^p_i}(\b\E)$ if and only if the states $f^k_j$ and $g^p_i$ of $\mathcal{B}_{T_{\overline{\E}}}$ are equivalent. We eliminate the $\varepsilon$-transitions from the automaton ${{\cal B}_{T_{_{\b\E}}}}$. Since it has no $\varepsilon$-transitions cycles, this elimination can be performed in a linear time in the size of $\E$. Hence, we obtain a more compacted but equivalent structure, which we denote by $\varepsilon\mbox{-}\mathrm{free}({{\cal B}_{T_{_{\b\E}}}})$.

\noindent \begin{minipage}{0.45\linewidth}
  \begin{figure}[H]
	 \centerline{
	  \begin{tikzpicture}[node distance=1cm,bend angle=30,transform shape,scale=0.75]
	    \node[state, initial below] (nut) {$\nu_T$};  
	    \node[state, above of=nut,node distance=2cm] (f22) {$f^2_2$};   	    \node[state, above left of=f22,node distance=1.5cm] (f2) {$f_2$};  
	    \node[state, below left of=f2,node distance=1.5cm] (f21) {$f^1_2$};   
	    \node[state, above left of=f2,node distance=2cm] (+) {$+$}; 
	    \node[state, below left of=+,node distance=2cm] (f1) {$f_1$}; 
	    \node[state, below right of=f1,node distance=1.5cm] (f12) {$f^2_1$}; 
	    \node[state, below left of=f1,node distance=1.5cm] (f11) {$f^1_1$}; 
	    \node[state, above of=+] (*a) {$*_a$}; 
	    \node[state, above of=nut,accepting,node distance=6.5cm] (pa) {$\cdot_a$};  
	    \node[state, right of=f22] (h31) {$h_3^1$}; 
	    \node[state, above of=h31] (h3) {$h_3$}; 
	    \path[->]
	      (nut) edge node[right] {$f^2_2$} (f22)
	      (nut) edge node[above right] {$f^1_2$} (f21)
	      (nut) edge[bend left] node[below left] {$f^2_1$} (f12)
	      (nut) edge[bend left] node[below left] {$f^1_1$} (f11)
	      (nut) edge[bend right] node[right] {$h_3^1$} (h31)
	      (h31) edge node[right] {$2$} (h3)
	      (f22) edge node[above right] {$1$} (f2)
	      (f21) edge node[above left] {$1$} (f2)
	      (f12) edge node[above right] {$1$} (f1)
	      (f11) edge node[above left] {$1$} (f1)
	      (f1) edge node[above left] {$\varepsilon$} (+)
	      (f2) edge node[above right] {$\varepsilon$} (+)
	      (+) edge node[left] {$\cdot_a 5$} (*a)
	      (*a) edge node[above left] {$\cdot_a 6$} (pa)
	      (h3) edge node[right] {$\varepsilon$} (pa)
	    ;
	  \end{tikzpicture}
	 }
	  \caption{The automaton ${\cal B}_{T_{_{\b\E}}}$.}
	  \label{fig aut btbe}
	\end{figure}
\end{minipage}
\hfil
\begin{minipage}{0.45\linewidth}
  \begin{figure}[H]
	 \centerline{
	  \begin{tikzpicture}[node distance=1cm,bend angle=30,transform shape,scale=0.75]
	    \node[state, initial below] (nut) {$\nu_T$};  
	    \node[state, above of=nut,node distance=2cm] (f22) {$f^2_2$};   	    \node[above left of=f22,node distance=1.5cm] (f2) {};  
	    \node[state, below left of=f2,node distance=1.5cm] (f21) {$f^1_2$};   
	    \node[state, above left of=f2,node distance=2cm] (+) {$+$}; 
	    \node[below left of=+,node distance=2cm] (f1) {}; 
	    \node[state, below right of=f1,node distance=1.5cm] (f12) {$f^2_1$}; 
	    \node[state, below left of=f1,node distance=1.5cm] (f11) {$f^1_1$}; 
	    \node[state, above of=+] (*a) {$*_a$}; 
	    \node[state, above of=nut,accepting,node distance=6.5cm] (pa) {$\cdot_a$};  
	    \node[state, right of=f22] (h31) {$h_3^1$}; 
	    \node[above of=h31] (h3) {}; 
	    \path[->]
	      (nut) edge node[right] {$f^2_2$} (f22)
	      (nut) edge node[above right] {$f^1_2$} (f21)
	      (nut) edge[bend left] node[below left] {$f^2_1$} (f12)
	      (nut) edge[bend left] node[below left] {$f^1_1$} (f11)
	      (nut) edge[bend right] node[right] {$h_3^1$} (h31)
	      (h31) edge node[right] {$2$} (pa)
	      (f22) edge node[above right] {$1$} (+)
	      (f21) edge node[right] {$1$} (+)
	      (f12) edge node[left] {$1$} (+)
	      (f11) edge node[above left] {$1$} (+)
	      (+) edge node[left] {$\cdot_a 5$} (*a)
	      (*a) edge node[above left] {$\cdot_a 6$} (pa)
	    ;
	  \end{tikzpicture}
	 }
	  \caption{The automaton $\varepsilon\mbox{-}\mathrm{free}({{\cal B}_{T_{_{\b\E}}}})$.}
	  \label{fig aut eps free btbe}
	\end{figure}
\end{minipage}

The computation of the equivalence relation $\sim_e$ can be performed by the computation of Myhill-Nerode relation~\cite{Nerode} on the states of the automaton $\varepsilon\mbox{-}\mathrm{free}({{\cal B}_{T_{_{\b\E}}}})$. This automaton is deterministic and acyclic.

\begin{theorem}\label{theo sim e temps lin}
The  relation $\sim_e$ can be computed in $O(|\E|)$ time complexity.
\end{theorem}
\begin{proof}
 
 The equivalence relation $\sim_e$ coincides with Myhill-Nerode equivalence~\cite{Nerode} on the states of the automaton 
 $\varepsilon\mbox{-}\mathrm{free}({{\cal B}_{T_{_{\b\E}}}})$. 
 
 This automaton is deterministic and acyclic and its size is linear with respect to $|\E|$ (Proposition~\ref{prop long pseudo cont lin}). That can be computed in $O(|\E|)$ time and space complexity using Revuz Algorithm~\cite{revuz}. 
  \qed
\end{proof}

\noindent \begin{minipage}{0.55\linewidth}
	\begin{example}
	  Let us consider the regular expression  $\E=(f(a,a)+f(a,a))^{*_a}\cdot_a h(b)$ of Example~\ref{ex calc c cont}.  The automaton ${\cal B}_{T_{_{\b\E}}}$ is represented by Figure~\ref{fig aut btbe}. The automaton $\varepsilon\mbox{-}\mathrm{free}({{\cal B}_{T_{_{\b\E}}}})$ is represented in Figure~\ref{fig aut eps free btbe}. Applying Myhill-Nerode equivalence to the automaton $\varepsilon\mbox{-}\mathrm{free}({{\cal B}_{T_{_{\b\E}}}})$ results in the automaton in Figure~\ref{fig min aut eps free btbe}. We deduce from this automaton that $C_{f^1_1}(\b\E)\sim_e C_{f^2_1}(\b\E)\sim_e C_{f^1_2}(\b\E)\sim_e C_{f^2_2}(\b\E)$. Consequently the set of states of ${\cal C_{\E}}\diagup_{\sim_e}$ is $\{[C_{\varepsilon^1}(\b E)],[C_{f^1_1}(\b\E)],[C_{h^1_3}(\E)]\}$.
	\end{example}
\end{minipage}
\hfill
\begin{minipage}{0.4\linewidth}
  \begin{figure}[H]
	 \centerline{
	  \begin{tikzpicture}[node distance=1cm,bend angle=30,transform shape,scale=0.75]
	    \node[state, initial below] (nut) {$\nu_T$};
	    \node[state, above left of=nut,rounded rectangle,node distance=2.5cm] (f) {$\{f_1^1,f_1^2,f_2^1,f_2^2\}$};
	    \node[state, above of=f] (+) {$+$};
	    \node[state, above of=+] (*a) {$*_a$};
	    \node[state, above of=nut, node distance=4.5cm] (pa) {$\cdot_a$};
	    \node[state, above right of=nut,rounded rectangle,node distance=2.5cm] (h) {$\{h_3^1\}$};
	    \path[->]
	      (nut) edge[bend left] node[below left] {$f_1^1,f_1^2,f_2^1,f_2^2$} (f)
	      (nut) edge[bend right] node[below right] {$h_3^1$} (h)
	      (f) edge node[left] {$1$} (+)
	      (+) edge node[left] {$\cdot_a 5$} (*a)
	      (*a) edge node[above left] {$\cdot_a 6$} (pa)
	      (h) edge[bend right] node[above right] {$2$} (pa)
	    ;
	  \end{tikzpicture}
	 }
	  \caption{The Minimal Automaton of $\varepsilon\mbox{-}\mathrm{free}({{\cal B}_{T_{_{\b\E}}}})$.}
	  \label{fig min aut eps free btbe}
	\end{figure}
\end{minipage}

\subsection{Computation of the set of transition rules}

Using Proposition~\ref{prop2} and Proposition\ref{prop equiv g-1 first}, we can show that the computation of the set of transitions of the equation tree automaton is performed by computing the function $\mathrm{Follow}$. The computation of a transition rule using Proposition~\ref{prop2} requires a linear time, according to Proposition~\ref{prop tps lin pour follow}. Then for all transition rules we get an $O(|Q\diagup_{\sim_e}| \times |\E|)$ time and space complexity where $Q$ is the set of $k-$C-Continuations of $\b\E$. 
The computation of the set of states $Q_{\b{\cal C}}\diagup_{\sim_e}$ make possible the creation of non-coaccessible states. Removing these states requires an $O(|Q_{\b{\cal C}}\diagup_{\sim_e}|\cdot|\E|)$ time complexity.

\begin{theorem}\label{theo eq tps cons}
 The equation tree automaton ${\cal A}_{\E}$ of  $E$ can be computed in $O(|Q|\cdot|\E|)$ time and space complexity with $Q$ the set of states of ${\cal A}_{\E}$ . 
\end{theorem}
\begin{proof}
  The equivalence relation $\sim_e$ can be computed in $O(|\E|)$ time and space complexity and the set of transition rules can be performed by 
  computing the function $\mathrm{Follow}$. The computation of a transition rule using Proposition~\ref{prop2} requires a linear time, according 
  to Proposition~\ref{prop tps lin pour follow}. Then for all transition rules we get an $O(|Q_{\b C}\diagup_{\sim_e}| \times |\E|)$ time and 
  space complexity where $Q_{\b C}$ is the set of $k-$C-Continuations of $\b\E$. Finally, removing not coaccessible states can be performed in linear time and results in the equation automaton.
  \qed
\end{proof}
\section{A Full Example}\label{sec:exemple}
Let $\E=h(h(c,b)\cdot_c a,a)\cdot_b (f(a,h(c,b))\cdot_c a+g(a))^{*_b}$ be a regular expression defined over the ranked alphabet $\Sigma$ 
such that ${\Sigma}^0=\{a,b,c\}$, $\Sigma^1=\{g\}$, $\Sigma^2=\{f,h\}$ and $\b\E=h_1(h_2(c,b)\cdot_c a,a)\cdot_b (f_3(a,h_4(c,b))\cdot_c a+g_5(a))^{*_b}$ be its linearized form.

The computation of the $k$-C-Continuations of the $\E$ using the Definition~\ref{def1} is given in Table~\ref{tab c-cont}.

\begin{table}[H]
\centerline{
    \begin{tabular}{l@{\ }l}
	    $C_{h^1_1}(\b\E)=$ & $(h_2(c,b)\cdot_c a)\cdot_b (f_3(a,h_4(c,b))\cdot_c a+g_5(a))^{*_b}$,\\
	    $C_{h^2_1}(\b\E)=$ & $a\cdot_b (f_3(a,h_4(c,b))\cdot_c a+g_5(a))^{*_b}$, \\
	    $C_{f^1_3}(\b\E)=$ & $a\cdot_c a\cdot_b (f_3(a,h_4(c,b))\cdot_c a+g_5(a))^{*_b}$,\\
		$C_{f^2_3}(\b\E)=$ & $(h_4(c,b)\cdot_c a)\cdot_b (f_3(a,h_4(c,b))\cdot_c a+g_5(a))^{*_b}$,\\ 
		$C_{h^1_2}(\b\E)=$ & $c\cdot_c a\cdot_b (f_3(a,h_4(c,b))\cdot_c a+g_5(a))^{*_b}$,\\    
		$C_{h^2_2}(\b\E)=$ & $b\cdot_c a\cdot_b (f_3(a,h_4(c,b))\cdot_c a+g_5(a))^{*_b}$,\\  
		$C_{h^1_4}(\b\E)=$ & $c\cdot_c a\cdot_b (f_3(a,h_4(c,b))\cdot_c a+g_5(a))^{*_b}$,\\ 
		$C_{h^2_4}(\b\E)=$ & $b\cdot_c a\cdot_b (f_3(a,h_4(c,b))\cdot_c a+g_5(a))^{*_b}$,\\ 
		$C_{g^1_5}(\b\E)=$ & $a\cdot_b (f_3(a,h_4(c,b))\cdot_c a+g_5(a))^{*_b}$.
   \end{tabular}
  }
  \caption{The $k$-C-Continuations of $\E$.}
  \label{tab c-cont}
\end{table}

From Table~\ref{tab c-cont}, the Follow function can be computed (Table~\ref{tabl folow}).

\begin{minipage}{0.4\linewidth}
\begin{table}[H]
  \centerline{
    \begin{tabular}{|@{\ }c@{\ }|@{\ }c@{\ }|}
      \hline
      $x^j_i$ & $Follow(E,x_i,j)$\\
      \hline
      $h^1_1$ & $\{h_2\}$\\
      $h^2_1$ & $\{a\}$\\
      $h^1_2$ & $\{a\}$\\
      $h^2_2$ & $\{g_5,f_3\}$\\
      $f^1_3$ & $\{a\}$\\
      $f^2_3$ & $\{h_4\}$\\
      $h^1_4$ & $\{a\}$\\
      $h^2_4$ & $\{g_5\}$\\
      $g^1_5$ & $\{a\}$\\
      \hline
    \end{tabular}
  }
  \caption{The function Follow.}
  \label{tabl folow}
\end{table}
\end{minipage}
\hfill
\begin{minipage}{0.55\linewidth}
 \begin{figure}[H] 
	 \centerline{
	\begin{tikzpicture}[node distance=1cm,bend angle=30,transform shape,scale=0.85]
		\node[above] (1) {$\cdot_b$};
		\node[below left of=1] (2) {$h_1$}; 
		\node[below right of=1] (3) {$*_b$}; 
		\node[below right of=2] (4) {$a$};
		\node[below left of=2] (5) {$\cdot_c$};
		\node[below right of=5] (16) {$a$};
		\node[below left of=5] (17) {$h_2$};
		\node[below left of=17] (18) {$c$}; 
        \node[below right of=17] (19) {$b$};
		\node[below right of=3] (6) {$+$};
		\node[below right of=6] (7) {$g_5$}; \node[below right of=7] (9) {$a$};
	\node[below left of=6] (8) {$\cdot_c$}; \node[below right of=8](11){$a$};
	 \node[below left of=8] (10) {$f_3$};  \node[below left of=10] (12) {$a$}; 
	 \node[below right of=10] (13) {$h_4$}; \node[below left of=13] (14) {$c$}; 
	 \node[below right of=13] (15) {$b$};
	 \path[->]
	      (1) edge node[above left,pos=0.4] {} (2)
	      (1) edge node[above right,pos=0.4] {} (3)
	      (2) edge node[left,pos=0.4] {} (4)
	      (2) edge node[left,pos=0.4] {} (5)
	      (3) edge node[right,pos=0.4] {} (6)
	      (5) edge node[above left,pos=0.4] {} (16)
	      (5) edge node[above left,pos=0.4] {} (17)
	      (6) edge node[above left,pos=0.4] {} (7)
	      (6) edge node[above right,pos=0.4] {} (8)
	      (7) edge node[above right,pos=0.4] {} (9)
	      (8) edge node[above left,pos=0.4] {} (10)
	      (8) edge node[above right,pos=0.4]{} (11)
	      (10) edge node[above right,pos=0.4] {} (12)
	      (10) edge node[above left,pos=0.4] {} (13)
	      (13) edge node[above right,pos=0.4] {} (14)
	      (13) edge node[above left,pos=0.4] {} (15)
	      (17) edge node[left,pos=0.4] {} (18)
	      (17) edge node[left,pos=0.4] {} (19)
	     ;
	      \path[->,color=blue,dotted]
	      (17) edge[bend right] (16)
	      (2) edge[bend right] (3)
	      (10) edge[bend right] (11)
	      (6) edge[bend right=55] (3)
	    ;
	     \node[inner sep=0pt,draw, fit=(1), color=blue] (eq1) {};
	    \node[left of=eq1,node distance=0.5cm,color=blue] (eq1et) {$\nu_1$};
	     \node[inner sep=0pt,draw, fit=(2), color=blue] (eq2) {};
	    \node[left of=eq2,node distance=0.5cm,color=blue] (eq2et) {$\nu_2$};
	      \node[inner sep=0pt,draw, fit=(17), color=blue] (eq3) {};
	    \node[left of=eq3,node distance=0.5cm,color=blue] (eq3et) {$\nu_3$};	    
	    \node[inner sep=0pt,draw, fit=(16), color=blue] (eq8) {};
	    \node[left of=eq8,node distance=0.5cm,color=blue] (eq8et) {$\nu_4$};
	     \node[inner sep=0pt,draw, fit=(3), color=blue] (eq4) {};
	      \node[left of=eq4,node distance=0.5cm,color=blue] (eq4et) {$\nu_5$};
	     \node[inner sep=0pt,draw, fit=(6), color=blue] (eq6) {};
	    \node[left of=eq6,node distance=0.5cm,color=blue] (eq6et) {$\nu_6$};
	     \node[inner sep=0pt,draw, fit=(10), color=blue] (eq5) {};
	    \node[left of=eq5,node distance=0.5cm,color=blue] (eq5et) {$\nu_7$};
	    \node[inner sep=0pt,draw, fit=(11), color=blue] (eq7) {};
	    \node[left of=eq7,node distance=0.5cm,color=blue] (eq7et) {$\nu_8$};
	 \end{tikzpicture}
	 }
	  \caption{The ZPC-Structure of $\E$.}
	  \label{fig a t e1}
	\end{figure}
\end{minipage}

Finally, from Table~\ref{tabl folow}, the transition function of  $\mathcal{C}_{\E}$ is the following:

\centerline{
  \begin{tabular}{r@{\ }c@{\ }l@{\ \ \ \ \ \ \ }r@{\ }c@{\ }l}
    $h(C_{h^1_1}(\b\E),C_{h^2_1}(\b\E)) $ & $\rightarrow$ & $ C_{\varepsilon^1}(\b\E)$ & $h(C_{h^1_2}(\b\E),C_{h^2_2}(\b\E)) $ & $\rightarrow$ & $ C_{h^1_1}(\b\E)$\\
    $a $ & $\rightarrow $ & $C_{h^2_1}(\b\E)$ & $a $ & $\rightarrow $ & $C_{h^1_2}(\b\E)$\\
    $g(C_{g^1_5}(\b\E))$ & $\rightarrow$ & $ C_{h^2_2}(\b\E)$ & $f(C_{f^1_3}(\b\E),C_{f^2_3}(\b\E)) $ & $\rightarrow$ & $ C_{h^2_2}(\b\E)$\\
    $a $  & $\rightarrow$ & $ C_{f^1_3}(\b\E)$ & $h(C_{h^1_4}(\b\E),C_{h^2_4}(\b\E))$ & $ \rightarrow $ & $C_{f^2_3}(\b\E)$\\
    $a$ & $ \rightarrow $ & $C_{g^1_5}(\b\E)$ & $a $ & $\rightarrow $ & $C_{h^1_4}(\b\E)$\\
    $f(C_{f^1_3(\b\E)},C_{f^2_3(\b\E)}) $ & $\rightarrow$ & $ C_{h^2_4}(\b\E)$ & $g(C_{g^1_5}(\b\E)) $ & $\rightarrow$ & $ C_{h^2_4}(\b\E)$\\
  \end{tabular}
}

  The ZPC-structure associated to $\E$ is represented in Figure~\ref{fig a t e1}. The dotted links in Figure~\ref{fig a t e1} represent the function $\gamma$:

\centerline{
  $\gamma(\nu_2)= \nu_5$, $\gamma(\nu_3)= \nu_4$, $\gamma(\nu_6)=\nu_5$, $\gamma(\nu_7)=\nu_8$.
  }
 
The automaton ${\cal A}_{T_{_{\E}}}$ associated with $\E$ is  represented in Figure~\ref{fig a t e2}. 
 
 \begin{figure}[H] 
	 \centerline{
		\begin{tikzpicture}
	[node distance=1cm,bend angle=30,transform shape,scale=1.15]
		\node[state,initial above] (1) {$\cdot_b$};
		\node[state,below left of=1,xshift=-5mm] (2) {$h_1$}; 
		\node[state,below right of=1,xshift=5mm] (3) {$*_b$}; 
		\node[state,below right of=2] (4) {$a$};
		\node[state,below left of=2] (5) {$\cdot_c$};
		\node[state,below right of=5] (16) {$a$};
		\node[state,below left of=5] (17) {$h_2$};
		\node[state,below left of=17] (18) {$c$}; 
	 \node[state,below right of=17] (19) {$b$};
		\node[state,below right of=3] (6) {$+$};
		\node[state,below right of=6] (7) {$g_5$}; \node[state,below right of=7] (9) {$a$};
	\node[state,below left of=6] (8) {$\cdot_c$}; \node[state,below right of=8](11){$a$};
	 \node[state,below left of=8] (10) {$f_3$};  \node[state,below left of=10] (12) {$a$}; 
	 \node[state,below right of=10] (13) {$h_4$}; \node[state,below left of=13] (14) {$c$}; 
	 \node[state,below right of=13] (15) {$b$};
	 \node[accepting,state,below of=14,node distance=1.5cm] (20) {$\nu_T$};
	 \path[->]
	      (1) edge node[above left,pos=0.4] {$g_{\cdot_b}$} (2)
	      (1) edge node[above right,pos=0.4] {$d_{\cdot_b}$} (3)
	      (2) edge node[right,pos=0.1] {$h^2$} (4)
	      (2) edge node[left,pos=0.1] {$h^1$} (5)
	      (3) edge node[right,pos=0.1] {$*^b$} (6)
	      (5) edge node[above right,pos=0.4] {$d_{\cdot_c}$} (16)
	      (5) edge node[above left,pos=0.4] {$g_{\cdot_c}$} (17)
	      (6) edge node[above right,pos=0.4] {$d_+$} (7)
	      (6) edge node[above left ,pos=0.4] {$g_+$} (8)
	      (7) edge node[above right,pos=0.4] {$g^1$} (9)
	      (8) edge node[above left,pos=0.6] {$g_{\cdot_c}$} (10)
	      (8) edge node[above right,pos=0.6]{$d_{\cdot_c}$} (11)
	      (10) edge node[above left,pos=0.4] {$f^1$} (12)
	      (10) edge node[above right,pos=0.4] {$f^2$} (13)
	      (13) edge node[above left,pos=0.4] {$h^1$} (14)
	      (13) edge node[above right,pos=0.4] {$h^2$} (15)
	      (17) edge node[left,pos=0.1] {$h^1$} (18)
	      (17) edge node[right,pos=0.1] {$h^2$} (19)
	      (4) edge [bend right]node[left,pos=0.45] {a} (20)
	      (9) edge [bend left=45] node[left,pos=0.6] {a} (20) 
	      (11) edge [bend left=55]node[left,pos=0.6] {a} (20)
	      (12) edge node[left,pos=0.4] {a} (20)
	      (14) edge node[left,pos=0.4] {c} (20)
	      (15) edge node[left,pos=0.4] {b} (20)
	      (16) edge [bend right]node[left,pos=0.4] {a} (20)
	      (18) edge [bend right]node[left,pos=0.4] {c} (20)
	      (19) edge [bend right]node[left,pos=0.4] {b} (20)	      
	     ;
	 \end{tikzpicture}
	 }
	  \caption{The automaton ${\cal A}_{T_{_{\E}}}$.}
	  \label{fig a t e2}
	\end{figure}	
	Applying Myhill-Nerode equivalence relation over the states of the automaton ${\cal A}_{T_{_{\E}}}$ results in the automaton in Figure~\ref{fig a t e3}.
		 
 \begin{figure}[H] 
	 \centerline{
		\begin{tikzpicture}
	[node distance=1cm,bend angle=50,transform shape,scale=1.15]
		\node[state,above] (1) {$g_5$};
		\node[state,below left of=1] (2) {$+$}; 
		\node[state,below left of=2] (3) {$*_b$}; 
		\node[state,below right of=2](4) {$\cdot_c$}; 
		\node[initial, state,below left of=3] (5) {$\cdot_b$};
		\node[state,below right of=5](6) {$h_1$};		
		\node[state,below right of=6](7) {$\cdot_c$};
		\node[state,below right of=4](8) {$f_3$};
		\node[state,below right of=8,node distance=1.5cm,scale=0.7,rounded rectangle] (9) {$\{h_2,h_4\}$};
		\node[state,above right of=9,node distance=1.5cm] (10) {$c$}; 
		\node[state,below right of=9,node distance=1.5cm](11) {$b$}; 
     	\node[accepting,state,below right of=10,node distance=1.5cm] (12) {$\nu_T$};
		\node[state,above right of=12,node distance=1.5cm,yshift=-11mm] (13) {$a$};
	 \path[->]
	     (5) edge node[above left,pos=0.4] {$d_{\cdot_b}$} (3)
         (3) edge node[above left] {$*_b$} (2)	     
	     (2) edge node[below right,pos=0.4] {$d_+$} (1)
	     (2) edge node[below left,pos=0.4] {$g_+$} (4)
         (4) edge node[below left,pos=0.4] {$g_{\cdot_c}$} (8)
         (8) edge node[below left,pos=0.4] {$f^2$} (9)
	     (5) edge node[below left,pos=0.4] {$g_{\cdot_b}$} (6)
	     (6) edge node[above right,pos=0.4] {$h_1$} (7)
	     (7) edge node[above left,pos=0.4] {$g_{\cdot_c}$} (9)
	     (9) edge node[above left,pos=0.4] {$h^1$} (10)
	     (9) edge node[below left,pos=0.4] {$h^2$} (11)
	      (10) edge node[below left,pos=0.4] {$c$} (12)
	      (11) edge node[above left,pos=0.4] {$b$} (12)
	      (13) edge node[above left,pos=0.4] {$a$} (12)
	      (1) edge [bend left=60] node[right,pos=0.8] {$g^1$} (13)
	      (8) edge [bend left=65] node[right,pos=0.6] {$~~f^1$} (13)
    	      (8) edge [bend left=35] node[right,pos=0.7] {$~~f^2$} (13)
    	      (7) edge [bend right=45] node[left,pos=0.2] {$d_{\cdot_c}$} (13)
    	      (6) edge [bend right=75] node[left,pos=0.18] {$h^2$} (13)	      
	     ;
	 \end{tikzpicture}
	 }
	  \caption{The minimal automaton of ${\cal A}_{T_{_{\E}}}$.}
	  \label{fig a t e3}
	\end{figure}	
	The computation of the equivalence relation $\sim$ over the syntax tree associated to $\E$ is represented in the Figure~\ref{fig a t e4}. The number of equivalence classes in Figure~\ref{fig a t e4} ($12$) corresponds exactly to the number of states of the minimal automaton of ${\cal A}_{T_{_{\E}}}$. From these equivalence classes, we can define the $\psi$ function (see Table~\ref{table psi}).  
	
	\begin{minipage}{0.40\linewidth}
	\begin{table}[H]
	  \centerline{
	    \begin{tabular}{l@{\ }l}
	      $\psi(1)$ & $=a$\\
	      $\psi(2)$ & $=b$\\
	      $\psi(3)$ & $=c$\\
	      $\psi(4)$ & $=h(c,b)$\\
	      $\psi(5)$ & $=g(a)$\\
	      $\psi(6)$ & $=h(c,b)\cdot_c a$\\
	      $\psi(7)$ & $=h(h(c,b)\cdot_c a,a)$\\
	      $\psi(8)$ & $=f(a,h(c,b))$\\
	      $\psi(9)$ & $=f(a,h(c,b)) \cdot_c a$\\
	      $\psi(10)$ & $=f(a,h(c,b)) \cdot_c a + g(a)$\\
	      $\psi(11)$ & $=(f(a,h(c,b)) \cdot_c a + g(a))^{*_b}$\\
	      $\psi(12)$ & $=E$\\
	    \end{tabular}
	  }
	  \caption{The function $\psi$.}
	  \label{table psi}
	\end{table}
	\end{minipage}
	\hfill
	\begin{minipage}{0.55\linewidth}
	\begin{figure}[H] 
	 \centerline{
	\begin{tikzpicture}[node distance=1cm,bend angle=30,transform shape,scale=0.8]
		\node(1) {$\cdot_b$};
		\node[below left of=1] (2) {$h_1$}; 
		\node[below right of=1] (3) {$*_b$}; 
		\node[below right of=2] (4) {$a$};
		\node[below left of=2] (5) {$\cdot_c$};
		\node[below right of=5] (16) {$a$};
		\node[below left of=5] (17) {$h_2$};
		\node[below left of=17] (18) {$c$}; 
	 \node[below right of=17] (19) {$b$};
		\node[below right of=3] (6) {$+$};
		\node[below right of=6] (7) {$g_5$}; \node[below right of=7] (9) {$a$};
	\node[below left of=6] (8) {$\cdot_c$}; \node[below right of=8](11){$a$};
	 \node[below left of=8] (10) {$f_3$};  \node[below left of=10] (12) {$a$}; 
	 \node[below right of=10] (13) {$h_4$}; \node[below left of=13] (14) {$c$}; 
	 \node[below right of=13] (15) {$b$};	 
	 \node[inner sep=0pt,draw, rectangle, fit=(1), color=red] (eq1) {};
	    \node[left of=eq1,node distance=0.5cm,color=red] (eq1et) {$12$};
	    \node[inner sep=0pt,draw, rectangle, fit=(2), color=red] (eq2) {};
	    \node[left of=eq2,node distance=0.5cm,color=red] (eq2et) {$7$};
	    \node[inner sep=0pt,draw, rectangle, fit=(3), color=red] (eq3) {};
	    \node[right of=eq3,node distance=0.5cm,color=red] (eq3et) {$11$};
	    \node[inner sep=0pt,draw, rectangle, fit=(5), color=red] (eq4) {};
	    \node[left of=eq4,node distance=0.4cm,color=red] (eq4et) {$6$};
	    \node[inner sep=0pt,draw, rectangle, fit=(4), color=red] (eq5) {};
	    \node[right of=eq5,node distance=0.3cm,color=red] (eq5et) {$1$};
	    \node[inner sep=0pt,draw, rectangle, fit=(16), color=red] (eq6) {};
	    \node[right of=eq6,node distance=0.3cm,color=red] (eq6et) {$1$};
	    \node[inner sep=0pt,draw, rectangle, fit=(19), color=red] (eq7) {};
	    \node[right of=eq7,node distance=0.3cm,color=red] (eq7et) {$2$};
	    \node[inner sep=0pt,draw, rectangle, fit=(18), color=red] (eq8) {};
	    \node[left of=eq8,node distance=0.3cm,color=red] (eq8et) {$3$};	
	    \node[inner sep=0pt,draw, rectangle, fit=(17), color=red] (eq9) {};
	    \node[left of=eq9,node distance=0.5cm,color=red] (eq9et) {$4$};    
	    \node[inner sep=0pt,draw, rectangle, fit=(6), color=red] (eq10) {};
	   \node[right of=eq10,node distance=0.5cm,color=red] (eq10et) {$10$};
	   \node[inner sep=0pt,draw, rectangle, fit=(7), color=red] (eq11) {};
	   \node[right of=eq11,node distance=0.5cm,color=red] (eq11et) {$5$};
	   \node[inner sep=0pt,draw, rectangle, fit=(9), color=red] (eq12) {};
	   \node[right of=eq12,node distance=0.3cm,color=red] (eq12et) {$1$};
	   \node[inner sep=0pt,draw, rectangle, fit=(11), color=red] (eq13) {};
	   \node[right of=eq13,node distance=0.3cm,color=red] (eq13et) {$1$};
	   \node[inner sep=0pt,draw, rectangle, fit=(15), color=red] (eq14) {};
	   \node[right of=eq14,node distance=0.3cm,color=red] (eq14et) {$2$};
	   \node[inner sep=0pt,draw, rectangle, fit=(12), color=red] (eq15) {};
	   \node[left of=eq15,node distance=0.3cm,color=red] (eq15et) {$1$};
	   \node[inner sep=0pt,draw, rectangle, fit=(14), color=red] (eq16) {};
	   \node[left of=eq16,node distance=0.3cm,color=red] (eq16et) {$3$};
	   \node[inner sep=0pt,draw, rectangle, fit=(8), color=red] (eq17) {};
	   \node[left of=eq17,node distance=0.5cm,color=red] (eq17et) {$9$};
	   \node[inner sep=0pt,draw, rectangle, fit=(10), color=red] (eq18) {};
	   \node[left of=eq18,node distance=0.4cm,color=red] (eq18et) {$8$};
	   \node[inner sep=0pt,draw, rectangle, fit=(13), color=red] (eq19) {};
	   \node[right of=eq19,node distance=0.4cm,color=red] (eq19et) {$4$};
	 \path[->]
	      (1) edge node[above left,pos=0.4] {} (2)
	      (1) edge node[above right,pos=0.4] {} (3)
	      (2) edge node[left,pos=0.4] {} (4)
	      (2) edge node[left,pos=0.4] {} (5)
	      (3) edge node[right,pos=0.4] {} (6)
	      (5) edge node[above left,pos=0.4] {} (16)
	      (5) edge node[above left,pos=0.4] {} (17)
	      (6) edge node[above left,pos=0.4] {} (7)
	      (6) edge node[above right,pos=0.4] {} (8)
	      (7) edge node[above right,pos=0.4] {} (9)
	      (8) edge node[above left,pos=0.4] {} (10)
	      (8) edge node[above right,pos=0.4]{} (11)
	      (10) edge node[above right,pos=0.4] {} (12)
	      (10) edge node[above left,pos=0.4] {} (13)
	      (13) edge node[above right,pos=0.4] {} (14)
	      (13) edge node[above left,pos=0.4] {} (15)
	      (17) edge node[left,pos=0.4] {} (18)
	      (17) edge node[left,pos=0.4] {} (19)
	     ;
	 \end{tikzpicture}
	 }
	  \caption{The equivalence classes.}
	  \label{fig a t e4}
	\end{figure}
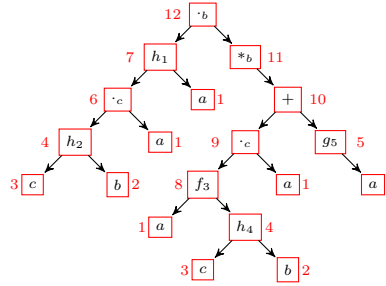
	\end{minipage}
	
	As we have seen, the computation of the equivalence relation $\sim_e$ turns in the minimization of an acyclic deterministic automaton. The automaton ${\cal B}_{T_{_{\b\E}}}$ associated with $\E$ is represented in Figure~\ref{fig a t e5}. 
	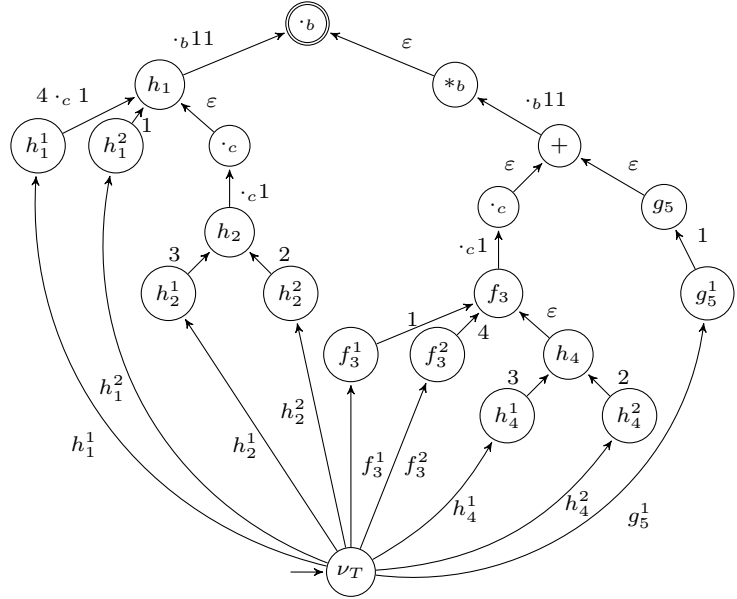
\begin{figure}[H]	
	 \centerline{
	\begin{tikzpicture}[node distance=1cm,bend angle=30,transform shape,scale=1.15]	
	\node[accepting,state, above] (1) {$\cdot_b$};
	\node[state,below left of=1,xshift=-10mm] (2) {$h_1$}; 
		\node[state,below right of=1,xshift=10mm] (3) {$*_b$};
		\node[state,below right of=2,xshift=1mm] (5) {$\cdot_c$};
		\node[state,below left of=2,xshift=-7mm] (7) {$h^1_1$};
		\node[state,below left of=2,xshift=2mm] (8) {$h^2_1$};
		\node[state,below right of=3,xshift=5mm] (6) {$+$};
		\node[state,below of=5] (9) {$h_2$}; 
		\node[state,below left of=9] (10) {$h^1_2$}; 
		\node[state,below right of=9] (11) {$h^2_2$};
		\node[state,below left of=6] (12) {$\cdot_c$}; 
		\node[state,below of=12] (15) {$f_3$};	
		\node[state,below right of=6,xshift=5mm] (13) {$g_5$};	
		\node[state,below of=13,xshift=5mm] (14) {$g^1_5$};	
		\node[state,below right of=15,xshift=1mm] (16) {$h_4$};
		\node[state,below left of=15,xshift=-10mm] (17) {$f^1_3$};
		\node[state,below left of=15] (18) {$f^2_3$};
		\node[state,below left of=16] (19) {$h^1_4$}; 
		\node[state,below right of=16] (20) {$h^2_4$};
		\node[state,initial,below left of=19,xshift=-11mm,yshift=-11mm] (21) {$\nu_T$}; 
		\path[->]
	      (2) edge node[above left,pos=0.4] {$\cdot_b11$} (1)
	      (3) edge node[above right,pos=0.4] {$\varepsilon$} (1)
	      (5) edge node[above right] {$\varepsilon$} (2)
	      (8) edge node[right,pos=-0.1] {$1$} (2)
	      (7) edge node[above left] {$4\cdot_c1$} (2)
	      (6) edge node[above right,pos=0.4] {$\cdot_b 11$} (3)
	       (13) edge node[above right,pos=0.4] {$\varepsilon$} (6)
	        (14) edge node[above right,pos=0.4] {$1$} (13)
	        (12) edge node[above left,pos=0.4] {$\varepsilon$} (6)
	        (15) edge node[above left,pos=0.1] {$\cdot_c1$} (12)
	        (16) edge node[above right,pos=0.4] {$\varepsilon$} (15)
	        (17) edge node[above right,pos=0.2] {$1$} (15)
	        (18) edge node[above right,pos=-0.6] {$~~~4$} (15)
	        (9) edge node[above right,pos=-0.1] {$\cdot_c 1$} (5)
	         (10) edge node[above left,pos=0.1] {$3$} (9)
	          (11) edge node[above right,pos=0.1] {$2$} (9)
	     	(19) edge node[above left,pos=0.1] {$3$} (16)
	          (20) edge node[above right,pos=0.1] {$2$} (16)
	      (21) edge [bend left=40]node[left,pos=0.45] {$h^1_1~~$} (7)
	      (21) edge [bend left=40]node[left,pos=0.55] {$h^2_1$} (8)
	      (21) edge node[left,pos=0.45] {$h^1_2$} (10)
	      (21) edge node[left,pos=0.6] {$h^2_2$} (11) 
	       (21) edge node[right,pos=0.5] {$f^1_3$} (17)
	      (21) edge node[right,pos=0.5] {$f^2_3$} (18) 
	      (21) edge [bend right=15] node[right,pos=0.5] {$h^1_4~~$} (19)
	      (21) edge [bend right=25] node[right,pos=0.7] {$h^2_4~~~~$} (20) 
	      	      (21) edge [bend right=45] node[right,pos=0.5] {$~~~g^1_5$} (14) 	          
	     ;
	 \end{tikzpicture}
	 }
	  \caption{The automaton ${\cal B}_{T_{_{\b\E}}}$.}
	  \label{fig a t e5}
	\end{figure}	
We eliminate the $\varepsilon$-transitions from the automaton ${\cal B}_{T_{_{\b\E}}}$. Since this last has no $\varepsilon$-transitions cycles, this elimination can be performed in a linear time in the size of $\E$. Hence, we obtain a structure which we denote $\varepsilon\mbox{-}\mathrm{free}({{\cal B}_{T_{_{\b\E}}}})$.  
	\begin{figure}[H]	
	 \centerline{
	\begin{tikzpicture}[node distance=1cm,bend angle=30,transform shape,scale=1.15]	
	\node[accepting,state, above] (1) {$\cdot_b$};
	\node[state,below left of=1,xshift=-10mm] (2) {$h_1$}; 
		\node[state,below left of=2,xshift=-7mm] (7) {$h^1_1$};
		\node[state,below left of=2,xshift=2mm] (8) {$h^2_1$};
		\node[state,below right of=3,xshift=5mm] (6) {$+$};
		\node[state,below of=5] (9) {$h_2$}; 
		\node[state,below left of=9] (10) {$h^1_2$}; 
		\node[state,below right of=9] (11) {$h^2_2$};
		\node[state,below of=12] (15) {$f_3$};	
		\node[state,below of=13,xshift=5mm] (14) {$g^1_5$};	
		\node[state,below left of=15,xshift=-10mm] (17) {$f^1_3$};
		\node[state,below left of=15] (18) {$f^2_3$};
		\node[state,below left of=16] (19) {$h^1_4$}; 
		\node[state,below right of=16] (20) {$h^2_4$};
		\node[state,initial,below left of=19,xshift=-11mm,yshift=-11mm] (21) {$\nu_T$}; 
		\path[->]
	      (2) edge node[above left,pos=0.4] {$\cdot_b11$} (1)
	      (8) edge node[right,pos=-0.1] {$1$} (2)
	      (7) edge node[above left] {$4\cdot_c1$} (2)
	      (6) edge node[above right,pos=0.4] {$\cdot_b 11$} (1)
	        (14) edge node[above right,pos=0.4] {$1$} (6)
	        (15) edge node[above left,pos=0.4] {$\cdot_c1$} (6)
	        (17) edge node[above right,pos=0.2] {$1$} (15)
	        (18) edge node[above right,pos=-0.6] {$~~~4$} (15)
	        (9) edge node[above right,pos=0.4] {$\cdot_c 1$} (2)
	         (10) edge node[above left,pos=0.1] {$3$} (9)
	          (11) edge node[above right,pos=0.1] {$2$} (9)
	     	(19) edge node[above right,pos=0.1] {$3$} (15)
	          (20) edge node[above right,pos=0.1] {$2$} (15)
	      (21) edge [bend left=40]node[left,pos=0.45] {$h^1_1~~$} (7)
	      (21) edge [bend left=40]node[left,pos=0.55] {$h^2_1$} (8)
	      (21) edge node[left,pos=0.45] {$h^1_2$} (10)
	      (21) edge node[left,pos=0.6] {$h^2_2$} (11) 
	       (21) edge node[right,pos=0.5] {$f^1_3$} (17)
	      (21) edge node[right,pos=0.5] {$f^2_3$} (18) 
	      (21) edge [bend right=15] node[right,pos=0.5] {$h^1_4~~$} (19)
	      (21) edge [bend right=25] node[right,pos=0.7] {$h^2_4~~~~$} (20) 
	      	      (21) edge [bend right=45] node[right,pos=0.5] {$~~~g^1_5$} (14) 	          
	     ;
	 \end{tikzpicture}
	 }
	  \caption{The automaton $\varepsilon\mbox{-}\mathrm{free}({{\cal B}_{T_{_{\b\E}}}})$.}
	  \label{fig a t e6}
	\end{figure}

	The computation of the equivalence relation $\sim_e$ amounts to apply Myhill-Nerode relation on the states of the automaton $\varepsilon\mbox{-}\mathrm{free}({{\cal B}_{T_{_{\b\E}}}})$. The result is represented in Figure~\ref{fig a t e7}.
	\begin{figure}[H]	
	 \centerline{
	\begin{tikzpicture}[node distance=2.5cm,bend angle=30,transform shape,scale=1.15]	
    \node[state, above of=1,rounded rectangle,node distance=1.5cm,] (9) {$\{h^1_1,f^2_3\}$};	
	\node[state, above] (1) {$f^1_3$};
	\node[initial,state,left of=1] (2) {$\nu_T$}; 
	\node[state, below of=1,rounded rectangle,node distance=1.5cm,] (3) {$\{h^1_2,h^1_4\}$};
	\node[state, below of=3,rounded rectangle,node distance=1.5cm,] (4) {$\{h^2_2,h^2_4\}$};
	\node[state, below of=4,rounded rectangle,node distance=1.5cm,] (5) {$\{h^2_1,g^1_5\}$};
	\node[state, right of=1,rounded rectangle] (6) {$\{h_2,f_3\}$};
	\node[state, right of=6,rounded rectangle] (7) {$\{h_1,+\}$};
	\node[accepting,state,right of=7] (8) {$\cdot_b$}; 
		\path[->]
	      (2) edge node[above left,pos=0.8] {$h^1_1,f^2_3$} (9)
	      (2) edge node[above left,pos=0.98] {$f^1_3$} (1)
	      (1) edge node[above right,pos=0.3] {$1$} (6)
	      (2) edge node[above right] {$h^1_2,h^1_4$} (3)
	      (2) edge node[ right,pos=0.75] {$h^2_2,h^2_4$} (4)
	      (2) edge node[below left,pos=0.4] {$h^2_1,g^1_5$} (5)	
	      (3) edge node[above ] {$3$} (6)
	      (4) edge node[below right,pos=0.5] {$2$} (6) 
	      (6) edge node[above left,pos=0.7] {$\cdot_c 1$} (7) 
	      (9) edge node[above right,pos=0.7] {$4\cdot_c 1$} (7)
	      (5) edge node[below] {$1$} (7) 
	      (7) edge node[above left,pos=0.9] {$\cdot_b11$} (8)      	          
	     ;
	 \end{tikzpicture}
	 }
	  \caption{The Minimal Automaton of $\varepsilon\mbox{-}\mathrm{free}({{\cal B}_{T_{_{\b\E}}}})$.}
	  \label{fig a t e7}
	\end{figure}		

%

The language recognized by $\mathcal{B}_{T_{\b{\E}}}$ is the following:

\centerline{
  \begin{tabular}{l@{\ }l@{\ }l}
    $\mathcal{L}(\mathcal{B}_{T_{\b{\E}}})=$ & & $\{h^2_2,h^2_4\}\cdot \{2\cdot_c 1\cdot_b11\}$\\
    & $\cup$ & $\{h^1_2,h^1_4\}\cdot \{3\cdot_c 1\cdot_b 11\}$\\
    & $\cup$ & $\{h^1_1,f^2_3\}\cdot \{4\cdot_c 1\cdot_b11\}$\\
    & $\cup$ & $\{h^2_1,g^1_5\}\cdot \{1\cdot_b11\}$\\
    & $\cup$ & $\{f^1_3 1\cdot_c 1\cdot_b11\}$\\
  \end{tabular}
}
Let us notice that Proposition~\ref{prop1} is satisfied in Table~\ref{tab lang ccont}.

	\begin{table}[H]
	\centerline{ 
 \begin{tabular}{|@{\ }c@{\ }|@{\ }c@{\ }|@{\ }c@{\ }|}
   \hline
   $x$ & $xw\in \mathcal{L}(\mathcal{B}_{T_{\b{\E}}})$ & $C_x(\overline{E})$\\
   \hline
   $h^1_1$ & $h^1_1 4\cdot_c 1 \cdot_b 11$ & $(h_2(c,b)\cdot_c a)\cdot_b (f_3(a,h_4(c,b))\cdot_c a+g_5(a))^{*_b}$\\
   $h^2_1$ & $h^2_1 1 \cdot_b 11$ &  $a\cdot_b (f_3(a,h_4(c,b))\cdot_c a+g_5(a))^{*_b}$\\
   $h^1_2$ & $h^1_2 3 \cdot_c 1\cdot_b 11$ &  $c\cdot_c a\cdot_b (f_3(a,h_4(c,b))\cdot_c a+g_5(a))^{*_b}$\\
   $h^2_2$ & $h^2_2 2\cdot_c 1 \cdot_b 11$ & $b\cdot_c a\cdot_b (f_3(a,h_4(c,b))\cdot_c a+g_5(a))^{*_b}$\\
   $f^1_3$ & $f^1_3 1 \cdot_c 1 \cdot_b 11$ & $a\cdot_c a\cdot_b (f_3(a,h_4(c,b))\cdot_c a+g_5(a))^{*_b}$\\
   $f^2_3$ & $f^2_3 4 \cdot_c 1\cdot_b 11$ & $(h_2(c,b)\cdot_c a)\cdot_b (f_3(a,h_4(c,b))\cdot_c a+g_5(a))^{*_b}$\\
   $h^1_4$ & $h^1_4 3\cdot_c 1\cdot_b 11$ & $c\cdot_c a\cdot_b (f_3(a,h_4(c,b))\cdot_c a+g_5(a))^{*_b}$\\ 
   $h^2_4$ & $h^2_4 2\cdot_c 1 \cdot_b 11$ & $b\cdot_c a\cdot_b (f_3(a,h_4(c,b))\cdot_c a+g_5(a))^{*_b}$\\
   $g^1_5$ & $g^1_5 1\cdot_b 11$ & $a\cdot_b (f_3(a,h_4(c,b))\cdot_c a+g_5(a))^{*_b}$\\ 
   \hline
   \end{tabular}
  }
    \caption{$\mathcal{L}(\mathcal{B}_{T_{\b{\E}}})$ and $k$-C-Continuations.}
    \label{tab lang ccont}
  \end{table}

  Finally, the equation automaton $\mathcal{A}_{\E}$ associated with $\E$ is obtained from merging the states and the transitions using $\sim_e$. The transition function is:
  
  \centerline{ 
  \begin{tabular}{r@{\ }c@{\ }l}
    $h(\{C_{h^1_1}(\b\E),C_{f^2_3}(\b\E)\},\{C_{h^2_1}(\b\E),C_{g^1_5}(\b\E)\}) $ & $\rightarrow$ & $ C_{\varepsilon^1}(\b\E)$\\
    $a $ & $\rightarrow $ & $\{C_{h^2_1}(\b\E),C_{g^1_5}(\b\E)\}$\\
     $h(\{C_{h^1_2}(\b\E),C_{h^1_4}(\b\E)\},\{C_{h^2_2}(\b\E),C_{h^2_4}(\b\E)\}) $ & $\rightarrow$ & $ \{C_{h^1_1}(\b\E),C_{f^2_3}(\b\E)\}$\\
     $a $ & $\rightarrow $ & $\{C_{h^1_2}(\b\E),C_{h^1_4}(\b\E)\}$\\
    $g(\{C_{g^1_5}(\b\E),C_{h^2_1}(\b\E)\})$ & $\rightarrow$ & $ \{C_{h^2_2}(\b\E),C_{h^2_4}(\b\E)\}$\\
    $f(\{C_{f^1_3}(\b\E)\},\{C_{f^2_3}(\b\E),C_{h^1_1}(\b\E)\}) $ & $\rightarrow$ & $ \{C_{h^2_2}(\b\E),C_{h^2_4}(\b\E)\}$\\
    $a $  & $\rightarrow$ & $ \{C_{f^1_3}(\b\E)\}$\\
  \end{tabular}
}

\section{Conclusion}

We presented a new and more efficient algorithm for the computation of the equation tree automaton from a regular tree expression by extending the notion of $k$-c-continuation from words to trees. We proved that a regular tree expression $\E$ can be converted into an equation tree automaton with an $O(|Q_{\b{\cal C}}\diagup_{\sim_e}|\cdot|\E|))$ time and space complexity 
where $Q$ is the set of $k$-C-Continuations of $\E$.

\bibliographystyle{splncs_srt}
\bibliography{bibliograph}

\begin{thebibliography}{10}

\bibitem{antimirov}
Antimirov, V.:
\newblock Partial derivatives of regular expressions and finite automaton
  constructions.
\newblock Theoretical computer Science \textbf{155} (1996)  291--319

\bibitem{Brug}
Bruggemann-Klein, A.:
\newblock Regular expressions into finite automata.
\newblock Theoretical computer Science \textbf{120} (1993)  197--213

\bibitem{ZPC1}
Champarnaud, J.M., Ziadi, D.:
\newblock From c-continuations to new quadratic algorithms for automaton
  synthesis.
\newblock Intern. J. of Algebra and Computation \textbf{11(6)} (2001)  707--735

\bibitem{ZPC2}
Champarnaud, J.M., Ziadi, D.:
\newblock Canonical derivatives, partial derivatives and finite automaton
  constructions.
\newblock Theoretical Computer Science \textbf{289(1)} (2002)  137--163

\bibitem{automate1}
Comon, H., Dauchet, M., Gilleron, R., Jacquemard, F., Lugiez, D., Loding, C.,
  Tison, S., Tommasi, M.:
\newblock Tree automata techniques and applications.
\newblock Available on: {\url{http://www.grappa.univ-lille3.fr/tata}} (October
  2007)

\bibitem{glushkov}
Glushkov, V.M.:
\newblock The abstract theory of automata.
\newblock Russian Mathematical Surveys \textbf{16} (1961)  1--53

\bibitem{khorsi}
Khorsi, A., Ouardi, F., Ziadi, D.:
\newblock Fast equation automaton computation.
\newblock Journal of Discrete Algorithms \textbf{6} (2008)  433--448

\bibitem{automate2}
Kuske, D., Meinecke, I.:
\newblock Construction of tree automata from regular expressions.
\newblock RAIRO - Theoretical Informatics and Applications \textbf{45} (2011)
  347--370

\bibitem{lata}
Laugerotte, E., Ouali-Sebti, N., Ziadi, D.:
\newblock From regular tree expression to position tree automaton.
\newblock Lecture Notes in Computer Science \textbf{7810} (2013)  395--406

\bibitem{xml}
Murata, M.:
\newblock Hedge automata: a formal model for xml schemata.
\newblock Available on: {\url{http://www.xml.gr.jp/relax/hedge_nice.html}}
  (2000)

\bibitem{Nerode}
Nerode, A.:
\newblock Linear automata transformation.
\newblock Proc. Amer. Math. Soc. \textbf{9} (1958)  541--544

\bibitem{tarjan}
R.~Paige, R.T.:
\newblock Three partition refinement algorithms.
\newblock SIAM Journal on Computing \textbf{16 (6)} (1987)  973--989

\bibitem{revuz}
Revuz, D.:
\newblock Minimization of acyclic deterministic automata in linear time.
\newblock Theoretical Computer Science \textbf{92(1)} (1992)  181--189

\bibitem{verif}
Trakhtenbrot, B.:
\newblock Origins and metamorphoses of the trinity: Logic, nets, automata.
\newblock In Proceedings, Tenth Annual IEEE Symposium on Logic in Computer
  Science. IEEE Computer Society Press (June 1995)  26--29

\bibitem{ZPC}
Ziadi, D., Ponty, J.L., Champarnaud, J.M.:
\newblock Passage d'une expression rationnelle a un automate fini non
  deterministe.
\newblock Bulletin of the Belgian Mathematical Society - Simon Stevin
  \textbf{4} (1997)  177--203

\end{thebibliography}

\end{document}